\documentclass[noinfoline,11pt]{imsart}

\RequirePackage[OT1]{fontenc}
\RequirePackage{amsthm,amsmath}
\RequirePackage[numbers]{natbib}
\RequirePackage[colorlinks,citecolor=blue,urlcolor=blue]{hyperref}

\usepackage{verbatim}

\allowdisplaybreaks

\usepackage{fullpage}

\newcommand{\tr}{{\rm tr}}
\newcommand{\ownint}[4]{{\int_{#1}^{#2} \! #3 \, \mathrm{d}#4}}
\newcommand{\conv}{\mathrm{conv}}
\newcommand{\innprod}[2]{{\left\langle {#1},{#2}\right\rangle}}
\newcommand{\topt}[2]{\mathbf t_{#1}^{#2}}
\newcommand{\Tan}{\mathrm{Tan}}

\newcommand {\R}{\mathbb{R}}
\newcommand {\W}{\mathcal{W}}
\newcommand {\X}{\mathcal{H}}

\newcommand{\E}{\mathbb{E}}

\newcommand {\cov}{\textrm{Cov}}

\DeclareFontEncoding{FMS}{}{}
\DeclareFontSubstitution{FMS}{futm}{m}{n}
\DeclareFontEncoding{FMX}{}{}
\DeclareFontSubstitution{FMX}{futm}{m}{n}
\DeclareSymbolFont{fouriersymbols}{FMS}{futm}{m}{n}
\DeclareSymbolFont{fourierlargesymbols}{FMX}{futm}{m}{n}
\DeclareMathDelimiter{\hsnorm}{\mathord}{fouriersymbols}{152}{fourierlargesymbols}{147}

\newcommand{\hs}{\big\hsnorm}

\usepackage{subfigure}
\usepackage{layout}

\usepackage{dsfont}
\usepackage[mathscr]{eucal}
\usepackage[toc,page]{appendix}
\usepackage{mathrsfs}
\usepackage{color}
\usepackage{pifont}
\usepackage{bm}
\usepackage{latexsym}
\usepackage{amsfonts}
\usepackage{amssymb}
\usepackage{epsfig}
\usepackage{graphicx}
\usepackage{multirow}

\newtheorem{theorem}{Theorem}
\newtheorem{conjecture}[theorem]{Conjecture}
\newtheorem{proposition}[theorem]{Proposition}
\newtheorem{corollary}[theorem]{Corollary}
\newtheorem{lemma}[theorem]{Lemma}
\newtheorem{definition}[theorem]{Definition}

\newcommand{\transpose}{^{\top}}
\startlocaldefs
\numberwithin{equation}{section}
\theoremstyle{plain}

\usepackage[OT1]{fontenc}
\usepackage{amsthm,amsmath,graphicx,latexsym,amssymb,array,mathabx,mathrsfs}
\usepackage{natbib}
\usepackage[colorlinks,citecolor=blue,urlcolor=blue]{hyperref}
\usepackage{adjustbox,blindtext}

\begin{document}

\begin{frontmatter}

\title{Procrustes Metrics on Covariance Operators and Optimal Transportation of Gaussian Processes}

\runtitle{Procrustes Analysis and Optimal Coupling}

\begin{aug}
\author{\fnms{Valentina} \snm{Masarotto}\ead[label=e1]{valentina.masarotto@epfl.ch}}, 
\author{\fnms{Victor M.} \snm{Panaretos}\ead[label=e2]{victor.panaretos@epfl.ch}}
\and
\author{\fnms{Yoav} \snm{Zemel}\ead[label=e3]{yoav.zemel@epfl.ch}}

\runauthor{V.~Masarotto, V.M.~Panaretos \& Y.~Zemel}

\affiliation{Ecole Polytechnique F\'ed\'erale de Lausanne}

\address{Institut de Math\'ematiques\\
Ecole Polytechnique F\'ed\'erale de Lausanne\\
\printead{e1}, \printead*{e2}, \printead*{e3}}

\end{aug}

\begin{abstract} Covariance operators are fundamental in functional data analysis, providing the canonical means to analyse functional variation via the celebrated Karhunen--Lo\`eve expansion. These operators may themselves be subject to variation, for instance in contexts where multiple functional populations are to be compared. Statistical techniques to analyse such variation are intimately linked with the choice of metric on covariance operators, and the intrinsic infinite-dimensionality of these operators. In this paper, we describe the manifold geometry of the space of trace-class infinite-dimensional covariance operators and associated key statistical properties, under the recently proposed infinite-dimensional version of the Procrustes metric. We identify this space with that of centred Gaussian processes equipped with the Wasserstein metric of optimal transportation. The identification allows us to provide a complete description of those aspects of this manifold geometry that are important in terms of statistical inference,  and establish key properties of the Fr\'echet mean of a random sample of covariances, as well as generative models that are canonical for such metrics and link with the problem of registration of functional data. 
\end{abstract}

\end{frontmatter}

\tableofcontents

\newpage

\section{Introduction}

\subsection*{Background and Contributions}

Covariance operators play a central role in \emph{functional data analysis} (\citet{hsing2015theoretical}, \citet{ramsay2005springer}): nonparametric inference on the law of a stochastic process $X$ viewed as a random element of an infinite-dimensional separable Hilbert space $\mathcal{H}$ (most usually $L^2$ or some reproducing kernel Hilbert subspace thereof).  In particular, covariance operators serve as the \emph{canonical means} to study the variation of such random functions. Their spectrum provides a singular system separating the stochastic and functional fluctuations of $X$, allowing for optimal finite dimensional approximations and functional PCA via the Karhunen--Lo\`eve expansion. And, that same singular system arises as the natural means of regularisation for inference problems (such as regression and testing) which are ill-posed in infinite dimensions (\citet{panaretos2013cramer}, \citet{wang2016functional}).
 
There are natural statistical applications where covariances may be the main object of interest in themselves, and may present variation of their own. These typically occur in situations where several different ``populations" of functional data are considered, and there is strong reason to suspect that each population may present different structural characteristics. Each one of $K$ populations is modelled by a prototypical random function $X_k$, with mean function $\mu_k\in\mathcal{H}$ and covariance operator ${\Sigma}_k:\mathcal{H}\times\mathcal{H}\rightarrow\mathcal{H}$ and we are able to observe $N_k$ realisations from each population: $\{X^{i}_{k}:i=1,\dots,N_k; k=1,\dots,K\}$. Examples of such situations include the two (or potentially more) populations of DNA strands considered in \citet{panaretos_jasa}, \citet{panaretos_biometrika}, and \citet{tavakoli:2014}, resulting from different base pair composition of each DNA strand, but clearly extend to much wider contexts.

A classical problem is the case where it is assumed that the different populations differ in their mean structure, leading to what has become known as Functional Analysis of Variance (see \citet{zhang:2013} for an overview). This represents \emph{first-order variation} across populations, as it can be considered as a model of the form
$$X_{k}^{i}(t)=\mu(t)+\mu_k(t)+\varepsilon_i(t),$$
with $\varepsilon_i(t)$ being mean zero and covarying according to some ${\Sigma}$.

An intriguing further type of variation is \emph{second-order variation}, which occurs by assuming that the covariance operators vary across populations, ${\Sigma}_i\neq{\Sigma}_j$ for $i\neq j$. This type of variation is particularly relevant in functional data, as it represents qualitative differences in the smoothness and fluctuation properties of the different populations. Early contributions in this area were motivated through financial  and biophysical applications \citep{benko:2009,panaretos_jasa}. These led to a surge of methods and theory on \emph{second-order variation} of functional populations, in many directions: \citet{horvath2013test}, \citet{paparoditis2014bootstrap}, \citet{gabrys2010tests}, \citet{fremdt2013testing}, \citet{horvath2012inference}, \citet{jaruvskova2013testing}, \citet{coffey2011common}, \citet{kraus2014components}.

What is common to many of these approaches is that the second-order variation is, in a sense, \emph{linear}. That is, the covariance operators are imbedded in the space of Hilbert-Schmidt operators, and statistical inference is carried out with respect to the corresponding metric. This space is, of course, a Hilbert space, and thus methodology of this form can be roughly thought of as modelling the second order variation via \emph{linear perturbations} of an underlying covariance operator:
\[
{\Sigma}_k={\Sigma}+{E}_k.
\]
Here ${E}_k$ would be a random zero-mean self-adjoint trace-class operator, with spectral constraints to assure the positive-definiteness of the left hand side. Being a random trace-class self-adjoint operator, ${E}$ admits its own Karhunen-Lo\`eve expansion, and this is precisely what has been employed in order to extend the linear PCA inferential methods from the case of functions.  However, the restriction ${\Sigma}+{E}_k \succeq 0$  immediately shows that the Hilbert-Schmidt approach has unavoidable weaknesses, as it imbeds covariance operators in a larger linear space, whereas they are not closed under linear operations. Quite to the contrary, covariance operators are fundamentally constrained to obey nonlinear geometries, as they are characterised as the ``squares" of Hilbert-Schmidt class operators.

In the multivariate (finite dimensional) literature this problem has been long known, and well-studied, primarily due to its natural connections with: (1) the problem of \emph{diffusion tensor imaging} (see, e.g., \cite{alexander2005multiple}, \cite{schwartzman2008false}, \cite{dryden2009non}) where it is fundamental in problems of smoothing, clustering, extrapolation, and dimension reduction, to name only a few; and (2) the statistical theory of shape (\citet{dryden1998statistical}), where Gram matrices (by definition non-negative) encode the invariant characteristics of Euclidean configurations under Euclidean motions. Consequently, inference for populations of covariance operators has been investigated under a wide variety of possible geometries for the space of covariance matrices (see, e.g., \citet{dryden2009non} or \citet{schwartzman2006random} for an overview). However, many of these metrics are based on quantities that do not lend themselves directly for generalisation to infinite dimensional spaces (e.g., determinants, logarithms and inverses).

\citet{pigoli2014distances} were the first to make important progress in the direction of considering second-order variation in appropriate nonlinear spaces, motivated by the problem of cross-linguistic variation of phonetics in Romance languages (where the uttering of a short word is modelled as a random function). They paid particular attention to the generalisation of the so-called Procrustes size-and-shape metric (which we will call simply Procrustes metric henceforth, for tidiness), and derived some of its basic properties, with a view towards initiating a programme of non-Euclidean analysis of covariance operators.  In doing so, they (implicitly or explicitly) generated many further interesting research directions on the geometrical nature of this metric, its statistical interpretation, and the properties of Fr\'echet means with respect to this metric. 

The purpose of this paper is to address some of these questions, and further our understanding of the Procrustes metric and the induced statistical models and procedures, thus placing this new research direction in non-Euclidean statistics on a firm footing. The starting point is a relatively straightforward but quite consequential observation: that the Procurstes metric between two covariance operators on $\mathcal{H}$ coincides with the Wasserstein metric between two centred Gaussian processes on $\mathcal{H}$ endowed with those covariances, respectively (Proposition~\ref{prop:ProcWass}, Section~\ref{sec:distances}). This connection allows us to exploit the wealth of geometrical and analytical properties of optimal transportation, and contribute in two ways. On the one hand, by reviewing and collecting some important aspects of Wasserstein spaces, re-interpreted in the Procrustean context, we elucidate key geometrical (Section~\ref{sec:geometry}), topological (Section~\ref{sec:topological}), and computational (Section~\ref{sec:algorithms}) aspects of the space of covariances endowed with the Procrustes metric.  On the other hand, we establish new results related to existence/uniqueness/stability of Fr\'echet means of covariances with respect to the Procrustes metric (Sections~\ref{sec:existUnique}), tangent space principal component analysis and Gaussian multicoupling (Section~\ref{sec:pca}), and generative statistical models compatible with the Procrustes metric and linking with the problem of warping/registration in functional data analysis (Section~\ref{sec:generative}). We conclude by formulating a conjecture on the regularity of the Fr\'echet mean that could have important consequences on statistical inference (Conjecture~\ref{conj:regularity}), and by posing some additional questions for future reseach (Section~\ref{sec:future}). The next paragraph collects the notational conventions employed throughout the paper, while an ancillary section (Section~\ref{sec:aux}) collects some background technical results, for tidiness.

\subsection*{Notation}
Let $\mathcal{H}$ be a real separable Hilbert space with inner product $\langle\cdot,\cdot\rangle:\mathcal{H}\times\mathcal{H}\rightarrow\mathbb{R}$, and induced norm $\|\cdot\|:\mathcal{H}\to[0,\infty)$. Given a bounded linear operator $A:\mathcal{H}\rightarrow\mathcal{H}$, we will denote its trace (when defined) by $\tr  A$ or $\tr( A)$, its adjoint operator by $A^*$, its Moore--Penrose generalised inverse by $A^{-}$, and its inverse by $A^{-1}$, which in general is only defined on a subspace (often dense) of $\X$. The \emph{kernel} of $A$ will be denoted by $\mathrm{ker}(A)=\{v\in\mathcal H:Av=0\}$, and its \emph{range} will be denoted by $\mathrm{range}(A)=\{Av:v\in\mathcal{H}\}$. When $A$ is positive (meaning that it is self-adjoint and $\langle Av,v\rangle\ge0$ for all $v\in\mathcal H$), the unique positive operator whose square equals $A$ will be denoted by either $A^{1/2}$ or $\sqrt{A}$.  For any bounded operator $A$, $A^*A$ is positive.  The identity operator on $\mathcal{H}$ will be denoted by $\mathscr{I}$.  The operator, Hilbert--Schmidt and nuclear norms will respectively be
\[
\hs A\hs_{\infty}=\sup_{\|h\|=1}\|Ah\|,
\quad \hs A \hs_2 =\sqrt{\tr\left(A^*A\right)},
\quad  \hs A \hs_1=\tr\left(\sqrt{A^*A}\right).
\] 
It is well-known that 
\[
\hs A\hs_{\infty} \le
\hs A\hs_2 \le
\hs A\hs_{1}
\]
for any bounded linear operator $A$.  When they are all finite, we say that $A$ is \emph{nuclear} or \emph{trace-class}.      Covariance operators are well-known to be positive and trace-class.

For a pair of elements $f,g\in \mathcal{H}$, the tensor product $f\otimes g:\mathcal{H}\to\mathcal{H}$ is the linear operator defined by
\[
(f\otimes g)u=\langle g,u\rangle f,\qquad u\in \mathcal{H}.
\]
The same notation will be used to denote the tensor product between two operators, so that for operators $A$, $B$, and $G$, one has
\[
(A\otimes B)\,G
=\tr\left(B^*G\right)A.
\]
Henceforth, $\Sigma$ or $\Sigma_i$ will always denote covariance operators.

\section{Procrustes Matching and Optimal Transportation}
\label{sec:distances}
\subsection{The Procrustes Distance Between Non-Negative Matrices and Operators}\label{procrustes_distance}
In classical statistical shape analysis, one often wishes to compare objects in $\mathbb{R}^m$ modulo a symmetry group $G$. To this aim, one chooses a fixed number of $k$ homologous landmarks on each object, represented by $k\times m$ matrices $X_1$ and $X_2$, and contrasts them by the Hilbert-Schmidt (a.k.a. Frobenius) distance of $X_1$ to $X_2$, optimally \emph{matched} relative to the group $G$. This induces a distance on the orbits of $X_1$ and $X_2$ under the group $G$, the latter called the \emph{shapes} of $X_1$ and $X_2$, and usually denoted as $[X_1]$ and $[X_2]$. For instance, if $G$ is the group of rigid motions on $\mathbb{R}^m$, one centres the configurations (so that their column sums are zero) and considers the so-called Procrustes shape-and-size distance  $\min_{U:\,U\transpose U=I} \|X_1 - BX_2\|_2$  (\cite[Definition~4.13]{dryden1998statistical}), henceforth abbreviated to \emph{Procrustes distance}, for simplicity. This distance depends only on the Gram matrices $X_1X_1\transpose $, $X_2X_2\transpose$, which can be thought of as parametrising the shapes $[X_1]$ and $[X_2]$. Since Gram matrices are non-negative, \citet{dryden2009non} considered the Procrustes distance as a metric on covariances $\mathbb{R}^{k\times k}\ni S_1, S_2\succeq 0$,
\begin{equation}\label{finite_procrustes}
\Pi(S_1,S_2)
=\inf_{U: \, U\transpose U=I}\| {S_1}^{1/2} - {S_2}^{1/2}U\|_2.
\end{equation}
The unique non-negative matrix roots $S_i^{1/2}$ in \eqref{finite_procrustes} can be replaced by \emph{any} matrices $Y_i$ such that $S_i=Y_i
Y_i\transpose$, but the former is the canonical choice in the context of covariances (in shape analysis, the $Y_i$ are typically chosen via the Cholesky decomposition, and are thought of as representatives from the corresponding shape equivalence classes).

Covariance operators are trace-class and can be fundamentally seen as ``squares" of operators with finite Hilbert-Schmidt norm. 
In order to analyse linguistic data, \citet{pigoli2014distances} considered the generalisation of the Procrustes distance \eqref{finite_procrustes} to the infinite-dimensional space of covariance operators on the separable Hilbert space $L^2(0,1)$.  Their definition applies readily, though, to any separable Hilbert space $\mathcal{H}$, and we give this more general definition here:

\begin{definition}[Procrustes Metric on Covariance Operators]
For any pair of nuclear and non-negative linear operators $\Sigma_1,\Sigma_2:\mathcal{H}\times\mathcal{H}\rightarrow\mathcal{H}$ on the separable Hilbert space $\mathcal{H}$, we define the Procrustes metric as
\begin{equation}\label{infinite_procrustes}
\Pi(\Sigma_1,\Sigma_2)=
\inf_{U:\,U^* U=\mathscr{I}}\hs\Sigma_1^{1/2}-U\Sigma_2^{1/2}\hs_2,
\end{equation}
where $\{U:\,U^* U=\mathscr{I}\}$ is the set of unitary operators on $\mathcal{H}$.
\end{definition}
Their motivation was mainly the construction of a procedure for testing the equality of two covariance operators on the basis of samples from the underlying two populations, tailored to the curved geometry of the space of covariance operators (as opposed to procedures based on embedding covariances in the linear space of trace-class or Hilbert-Schmidt operators). \citet{pigoli2014distances} consider the behaviour of $\Pi$ when considering finite-dimensional projections of the operators under consideration with progressively increasing dimension, and construct a permutation-based test on the distance between the projections.  They also discuss interpolation, geodesic curves and Fr\'echet means in the space of covariance operators endowed with the distance $\Pi$. In the next three subsections, we show that the distance $\Pi$ can be interpreted as a Wasserstein distance $W$.  This observation will allow us not only to shed new light on the results of \citet{pigoli2014distances}, but also to give a more comprehensive description of the geometry of the space as well as to address some questions that were left open by \citet{pigoli2014distances}.


\subsection{The Wasserstein Distance and Optimal Coupling}
In this subsection, we recall the definition of the Wasserstein distance and review some of its properties that will be used in the paper; we follow \citet{villani2003topics}.  Let $\mu$ and $\nu$ be Borel probability measures on $\mathcal H$ and let $\Gamma(\mu,\nu)$ be the set couplings of $\mu$ and $\nu$.  These are Borel probability measures $\pi$ on $\mathcal H\times\mathcal H$ such that $\pi(E\times \mathcal H)=\mu(E)$ and $\pi(\mathcal H\times F)=\nu(F)$ for all Borel $E,F\subseteq\mathcal H$.  The Wasserstein distance between $\mu$ and $\nu$ is defined as
\[
W^2(\mu,\nu)
=\inf_{\pi\in \Gamma(\mu,\nu)}
\ownint{\mathcal H\times\mathcal H}{}{\|x-y\|^2}{\pi(x,y)}.
\]
The distance is finite when $\mu$ and $\nu$ have a finite second moment, meaning that they belong to the Wasserstein space
\[
\mathcal W(\mathcal H)
=\left\{\mu\textrm{ Borel probability measure on }\mathcal H:\ownint {\mathcal H}{}{\|x\|^2}{\mu(x)}
<\infty\right\}.
\]
This optimisation problem is known as the Monge--Kantorovich problem of optimal transportation, and admits a natural probabilistic formulation. Namely, if $X$ and $Y$ are random elements on $\mathcal H$ with respective probability laws $\mu$ and $\nu$, then the problem translates to the minimisation problem
\[
\inf_{Z_1\stackrel{d}{=}X,\,Z_2\stackrel{d}{=}Y}\mathbb E\|Z_1 - Z_2\|^2
\]
where the infimum is over all random vectors $(Z_1,Z_2)$ in $\mathcal H\times \mathcal H$ such that $X\stackrel{d}{=}Z_1$ and $Y\stackrel{d}{=}Z_2$, marginally.  We sometimes write $W(X,Y)$ instead of $W(\mu,\nu)$.  We say that a coupling $\pi$ is deterministic if it is manifested as the joint distribution of $(X,T(X))$ for some deterministic map $T:\mathcal H\to\mathcal H$, called an \emph{optimal transportation map} (or simply optimal map, for brevity).  In such a case $Y$ has the same distribution as $T(X)$ and we write $\nu=T\#\mu$ and say that $T$ pushes $\mu$ forward to $\nu$ .  If $\mathscr I$ is the identity map on $\mathcal H$, we can write $\pi$ in terms of $T$ as $\pi=(\mathscr I,T)\#\mu$, and we say that $\pi$ is induced from $T$. In order to highlight the fact that the optimal map $T$ transports $\mu$ onto $\nu$, \citet{ambrosio2008gradient} introduced the notation $T\equiv \mathbf{t}_{\mu}^{\nu}$, and we will make use of this notation henceforth.

A simple compactness argument shows that the infimum in the Monge-Kantorovich problem is always attained by some coupling $\pi$, for any marginal pair of measures $\mu,\nu\in \mathcal{W}(\mathcal{H})$.  Moreover, when $\mu$ is sufficiently regular\footnote{In finite dimensions, it suffices that $\mu$ be absolutely continuous with respect to Lebesgue measure.  In infinite dimensions, a Gaussian measure is regular if and only if its covariance operator is injective.  For a more general definition, see \citet[Definition~6.2.2]{ambrosio2008gradient}}, the optimal coupling is unique and given by a deterministic coupling $\pi=(\mathscr I,\mathbf{t}_{\mu}^{\nu})\#\mu$ (by symmetry, if $\nu$ is regular then the optimal coupling is unique too and takes the form $(\mathbf{t}_{\nu}^{\mu},\mathscr I)\#\nu$).

\subsection{Optimal Transportation of Gaussian Processes}\label{gauss_wasserstein_distance}
Despite admitting a useful characterisation as the gradient of a convex function (\citet{brenier1991polar}; \citet{cuesta1989notes}; \citet{knott1984optimal}; \citet{ruschendorf1990characterization}), the optimal transportation map $\mathbf{t}_{\mu}^{\nu}$ (and, consequently, the corresponding Wasserstein distance $W(\mu,\nu)=\sqrt{\ownint{\X}{}{\|x-\mathbf{t}_{\mu}^{\nu}(x)\|^2}{\mu(x)}}$) rarely admit closed-form expressions.  A notable exception is the case where $\mu$ and $\nu$ are Gaussian\footnote{Recall that a random element $X$ in a separable Hilbert space $(\mathcal{H},\langle\cdot,\cdot\rangle)$ is Gaussian with mean $m\in\mathcal{H}$ and covariance $\Sigma:\mathcal{H}\times\mathcal{H}$, if $\langle X,h\rangle\sim N(\langle m,h\rangle,\langle h,\Sigma h\rangle)$ for all $h\in\mathcal{H}$; a Gaussian measure is the law of a Gaussian random element.}. Suppose that $\mu\equiv N(m_1,\Sigma_1)$ and $\nu\equiv N(m_2,\Sigma_2)$ are Gaussian measures. Then
\[
W^2(\mu,\nu)
=\|m_1 - m_2\|^2
+\tr(\Sigma_1) + \tr(\Sigma_2) - 2\tr\sqrt{\Sigma_1^{1/2}\Sigma_2\Sigma_1^{1/2}}.
\]
This was shown by \citet{dowson1982frechet} and \citet{olkin1982distance} in the finite-dimensional case.  For a reference in separable Hilbert spaces, see \citet{cuesta1996lower}.

There is also an explicit expression for the optimal map, but its existence requires some regularity. To simplify the discussion, assume henceforth that the two Gaussian measures $\mu$ and $\nu$ are centered, i.e., $m_1=m_2=0$.  When $\mathcal H=\mathbb R^d$ is finite-dimensional, invertibility of $\Sigma_1$ guarantees the existence and uniqueness of a deterministic optimal coupling of $\mu\equiv N(0,\Sigma_1)$ of $\nu\equiv N(0,\Sigma_2)$, induced by the linear transport map
\[
\topt{\Sigma_1}{\Sigma_2}:=\Sigma_1^{-1/2}(\Sigma_1^{1/2}\Sigma_2\Sigma_1^{1/2})^{1/2}\Sigma_1^{-1/2}.
\]
This formula turns out to be (essentially) valid in infinite dimensional Hilbert spaces $\mathcal{H}$, provided that $\Sigma_1$ is ``more injective" than $\Sigma_2$, but the statement is a bit more subtle:
\begin{proposition}\label{prop:existenceMaps}
Let $\mu\equiv N(0,\Sigma_1)$ and $\nu\equiv N(0,\Sigma_2)$ be centred Gaussian measures in $\mathcal{H}$ and suppose that $\mathrm{ker}(\Sigma_1)\subseteq \mathrm{ker}(\Sigma_2)$ (equivalently, $\overline{\mathrm{range}(\Sigma_1)}\supseteq \overline{\mathrm{range}(\Sigma_2)}$).  Then there exists a linear subspace of $\mathcal{H}$ with $\mu$-measure 1, on which the optimal map is well-defined and is given by the linear operator
\[
\topt{\Sigma_1}{\Sigma_2}=\Sigma_1^{-1/2}(\Sigma_1^{1/2}\Sigma_2\Sigma_1^{1/2})^{1/2}\Sigma_1^{-1/2}.
\]
\end{proposition}
Proposition~\ref{prop:existenceMaps} is established by \citet[Proposition~2.2]{cuesta1996lower}.  The same reference also shows that $\mathrm{ker}(\Sigma_1)\subseteq \mathrm{ker}(\Sigma_2)$ is indeed a \emph{necessary} condition in order that the optimal map exist. In general, the linear map $\topt{\Sigma_1}{\Sigma_2}$ is an unbounded operator and \emph{cannot} be extended to the whole of $\mathcal H$. Note that we've used the obvious switch in notation $\topt{\Sigma_1}{\Sigma_2}$ in lieu of $\topt{\mu}{\nu}$ when $\mu\equiv N(0,\Sigma_1)$ and $\nu\equiv N(0,\Sigma_2)$.

In the special case where $\Sigma_1$ and $\Sigma_2$ commute ($\Sigma_1\Sigma_2=\Sigma_2\Sigma_1$), the proof of Proposition~\ref{prop:existenceMaps} is quite simple, and indeed instructive in highlighting the subtleties involved in infinite dimensions.  Assume without loss of generality that $\Sigma_1$ is injective (otherwise replace $\X$ by the closed range of $\Sigma_1$).  The domain of definition of $\Sigma_1^{-1/2}$ is the range of $\Sigma_1^{1/2}$, which is dense in $\X$;  however this range has $\mu$-measure zero.  The problem is compensated by the compactness of $\Sigma_2^{1/2}$.  Let $\{e_k\}$ be an orthonormal basis of $\X$ composed of the eigenvectors of $\Sigma_1$ and $\Sigma_2$ (they share the same eigenvectors, since they commute) with eigenvalues $a_k$ and $b_k$.  Then $\topt{\Sigma_1}{\Sigma_2}$ simplifies to $\Sigma_2^{1/2}\Sigma_1^{-1/2}$, and is defined for all $x=\sum x_ke_k\in\X$ such that
\[
\sum_{k=1}^\infty (x_kb_k^{1/2}/a_k^{1/2})^2
=\sum_{k=1}^\infty x_k^2b_k/a_k
\]
is finite.  If $X\sim N(0,\Sigma_1)$, then $X_k=\langle X,e_k\rangle$ are independent, and by Kolmogorov's Three Series Theorem (\citet[Theorem~2.5.4]{durrett2010probability}) the above series converges almost surely because $\E X_k^2=a_k$ and $\sum \E X_k^2b_k/a_k=\sum b_k=\tr\Sigma_2<\infty$ because $\Sigma_2$ is trace-class;  the other two series in the theorem are also easily verified to converge.  We see that $\topt{\Sigma_1}{\Sigma_2}$ is bounded if and only if $b_k/a_k$ is bounded, which may or may not be the case.

\subsection{Procrustes Covariance Distance and Gaussian Optimal Transportation}

We now connect the material in Subsections \ref{procrustes_distance}--\ref{gauss_wasserstein_distance}, to make the following observation:

\begin{proposition}\label{prop:ProcWass}
The Procrustes distance between two trace-class covariance operators $\Sigma_1$ and $\Sigma_2$ on $\mathcal{H}$ coincides with the Wasserstein distance between two second-order Gaussian processes $N(0,\Sigma_1)$ and $N(0,\Sigma_2)$ on $\mathcal{H}$,
\begin{align*}
\Pi(\Sigma_1,\Sigma_2)&= \inf_{R:\,R^* R=\mathscr{I}}\hs\Sigma_1^{1/2}-U\Sigma_2^{1/2}\hs_2\\
&= \sqrt{\tr(\Sigma_1) + \tr(\Sigma_2) - 2\tr\sqrt{\Sigma_2^{1/2}\Sigma_1\Sigma_2^{1/2}}}= W(N(0,\Sigma_1),N(0,\Sigma_2)).
\end{align*}
\end{proposition}

\begin{proof}
Following \citet{pigoli2014distances}, we write
\[
\Pi^2(\Sigma_1,\Sigma_2)
=\inf_{R} \tr[(\sqrt{\Sigma_1} - \sqrt{\Sigma_2}R)^*(\sqrt{\Sigma_1} - \sqrt{\Sigma_2}R)]
=\tr\Sigma_1 + \tr\Sigma_2 - 2\sup_{R} \tr(R^*\Sigma_2^{1/2}\Sigma_1^{1/2}).
\]
Let $C=[\Sigma_2^{1/2}\Sigma_1^{1/2}]^*\Sigma_2^{1/2}\Sigma_1^{1/2}=\Sigma_1^{1/2}\Sigma_2\Sigma_1^{1/2}$ and the singular value decomposition $\Sigma_2^{1/2}\Sigma_1^{1/2}=U C^{1/2} V$ for $U$ and $V$ unitary.  Then $\tr(R^*\Sigma_2^{1/2}\Sigma_1^{1/2})=\tr(VR^*UC^{1/2})$ is maximised when $VR^*U$ is the identity (since $\{VR^*U:R^*R=\mathscr{I}\}$ is precisely the collection of unitary operators, and $C^{1/2}$ is positive).  We thus have
\[
\Pi^2(\Sigma_1,\Sigma_2)
=\tr(\Sigma_1) + \tr(\Sigma_2) - 2\tr\left[\sqrt{\Sigma_2^{1/2}\Sigma_1\Sigma_2^{1/2}}\right].
\]
\end{proof}

It is worth point out that if $\Sigma_1$ and $\Sigma_2$ happen to commute, then the product $\Sigma_1^{1/2}\Sigma_2^{1/2}$ is self adjoint, so that  $\sqrt{(\Sigma_1^{1/2}\Sigma_2^{1/2})^*\Sigma_1^{1/2}\Sigma_2^{1/2}}=\Sigma_1^{1/2}\Sigma_2^{1/2}$ and the Wasserstein distance reduces to the Hilbert--Schmidt distance of the covariance roots:
\[
W^2(N(0,\Sigma_1),N(0,\Sigma_2))
=  \hs \Sigma_1^{1/2} \hs_2^2+\hs \Sigma_2^{1/2} \hs_2^2-2\langle \Sigma_1^{1/2}, \Sigma_2^{1/2} \rangle_{HS}=\hs \Sigma_1^{1/2} -\Sigma_2^{1/2}\hs_2^2,
\]
with the optimal map being $\Sigma_2^{1/2}\Sigma_1^{-1/2}$.

We shall now take advantage of the vast wealth of knowledge about optimal transportation theory in order to gain further insight on the geometry and the topology of the space of covariance operators endowed with the Procrustes metric.

\section{The Tangent Bundle}\label{sec:geometry}
In this section we review some results from the book of \citet{ambrosio2008gradient}, where it is shown how the Wasserstein distance $W$ induces a manifold geometry on the Wasserstein space $\mathcal W(\mathcal H)$.  We then translate these results into geometrical properties of the space of covariance operators, equipped with the Procrusted distance (by identifying the latter with the subspace of $\mathcal W(\mathcal H)$ that consists of centred Gaussian measures; see \citet{takatsu2011wasserstein} for a detailed description of this subspace in the finite dimensional case).  Let $\mu\in\mathcal W(\mathcal H)$ and introduce the $L_2$-like space and norm of Borel functions $f:\mathcal H\to\mathcal H$ by
\[
\|f\|_{\mathcal L_2(\mu)}
=\left(  \ownint {\mathcal H}{}{\|f(x)\|^2}{\mu(x)}  \right)^{1/2}
,
\qquad 
\mathcal L_2(\mu)
=\{f:\|f\|_{\mathcal L_2(\mu)}
<\infty\}.
\]


Let $\mu,\nu\in \mathcal W(\mathcal H)$ be such that the optimal map from $\mu$ to $\nu$, $\topt\mu\nu$, exists.  Recalling that $\mathscr I:\mathcal H\to\mathcal H$ is the identity map, we can define a curve
\[
\mu_t=\left[\mathscr I +t(\topt {\mu}{\nu} - \mathscr I)\right]
\#\mu,\qquad t\in[0,1].
\]
This curve, known as McCann's interpolation (\citet[Equation~7]{mccann1997convexity}, is a constant speed geodesic in that 
$\mu_0=\mu$, $\mu_1=\nu$ and
\[
W(\mu_t,\mu_s)
=(t-s)W(\mu,\nu),
\qquad 0\le s\le t\le 1.
\]

\noindent The {tangent space} of $\mathcal W(\mathcal H)$ at $\mu$ is (\citet[Definition 8.5.1]{ambrosio2008gradient})
\[
\Tan_\mu
=\overline{\{t(\mathbf t-\mathscr I):\mathbf t\textrm{ uniquely optimal between }\mu \textrm{ and }\mathbf t\#\mu; t>0\}}^{\mathcal L_2(\mu)}.
\]
Since $\mathbf t$ is uniquely optimal, $\mathbf t\#\mu\in \mathcal W(\mathcal H)$ as well and $x\mapsto\|\mathbf t(x)\|$ is in $\mathcal L_2(\mu)$, so $\Tan_\mu\subseteq \mathcal L_2(\mu)$. Since optimality of $\mathbf t$ is independent of $\mu$, the only part of this definition that depends on $\mu$ is the closure operation.  Although not obvious from the definition, this is a linear space.\footnote{There is an equivalent definition in terms of gradients, in which linearity is clear, see \cite[Definition~8.4.1]{ambrosio2008gradient}:  when $\mathcal H=\R^d$, it is $\Tan_\mu = \overline{\{\nabla f:f\in C_c^\infty(\R^d)\}}^{\mathcal L_2(\mu)}$ (compactly supported $C^\infty$ functions).  When $\mathcal H$ is a separable Hilbert space, one takes $C_c^\infty$ functions that depend on finitely many coordinates, called \textbf{cylindrical functions} \cite[Definition~5.1.11]{ambrosio2008gradient}.  The two definitions of the tangent space coincide by \cite[Theorem~8.5.1]{ambrosio2008gradient}.)}

The exponential map ${\exp}_\mu:\Tan_\mu \to \mathcal W(\X)$ at $\mu$ is given by
\[
{\exp}_{\mu}(t(\mathbf t - \mathscr I))
={\exp}_{\mu}([t\mathbf t + (1-t)\mathscr I] - \mathscr I)
= [t\mathbf t + (1-t)\mathscr I]\#\mu
\quad(t\in\R).
\]
It is surjective if $\mu$ is regular. Consequently, if $\mu$ is regular, the (right) inverse of the exponential map, the log map ${\log}_\mu:\mathcal W(\X) \to \Tan_\mu$, is well-defined defined throughout $\mathcal W(\X)$, and given by
\[
\log_{\mu}(\nu)
=\topt{\mu}{\nu} - \mathscr I.
\]
In particular, one has
\[
\exp_\mu(\log_\mu(\nu))=\nu
,\quad \nu\in \mathcal W,
\qquad \textrm{and}\qquad
\log_\mu(\exp_\mu(t(\mathbf t - \mathscr I)))
= t(\mathbf t - \mathscr I)
\quad(t\in[0,1]),
\]
because convex combinations of optimal maps are optimal maps as well, and so McCann's interpolant $\left[\mathscr I+t(\topt{\mu}{\nu} - \mathscr I)\right]\#\mu$ is mapped bijectively to the line segment $t(\topt\mu\nu - \mathscr I)\in \Tan_{\mu}$ through the log map.

Let us now translate this geometric discussion to the space of covariance operators equipped with the Procrustes metric $\Pi$ (by implicitly focussing on centred Gaussian measures in $\mathcal{W}(\mathcal{H})$).  In this case, writing $\Tan_\Sigma$ for $\Tan_{N(0,\Sigma)}$, a unique optimal map $\mathbf t$ is a positive, possibly unbounded operator such that $\mathbf t\Sigma\mathbf t$ is trace-class.  In other words, $\Sigma^{1/2}\mathbf t$ is Hilbert--Schmidt, which is equivalent to $\Sigma^{1/2}(\mathbf t- \mathscr I)$ being Hilbert--Schmidt.  We consequently obtain the description of the tangent space at $\Sigma$ as 
\[
\Tan_\Sigma
=\overline{\left\{
t(S-\mathscr I):  t>0,\, S\succeq 0,\, \hs \Sigma^{1/2} (S - \mathscr I)\hs_2<\infty\right\}}
=\overline{\left\{
Q:  Q=Q^*,\, \hs \Sigma^{1/2} Q\hs_2<\infty\right\}}
,
\]
where the closure is with respect to the inner product on $\Tan_\Sigma$, defined as
\begin{equation}\label{tangent_inner_product}
\langle A,B\rangle_{\Tan_\Sigma}
=\ownint{\mathcal H}{}{\langle Ax, Bx\rangle}{\mu(x)}
=\tr( A\Sigma B) = \mathbb{E}[\langle AX,BX \rangle]
,\qquad
\mbox{where }\, X\sim\mu\equiv N(0,\Sigma).
\end{equation}
(For the second equality in the definition of $\Tan_\Sigma$, notice that if $Q$ is a bounded self adjoint operator, then $S=\mathscr I+Q/t$ is positive when $t>\hs Q\hs_\infty$;  unbounded $Q$'s can then be approximated.)  When equipped with this inner product, $\Tan_\Sigma$ is a Hilbert space.  Note that $\Tan_\Sigma$ certainly contains all bounded self-adjoint operators on $\mathcal{H}$, but also certain unbounded ones.  For example, if $\Sigma^{1/3}$ is trace-class, then the tangent space inner product is well defined when taking $A=B=\Sigma^{-1/3}$, which is an unbounded operator.

The exponential map on $\Tan_\Sigma$ is given by $\exp_{\Sigma}(A)=(A+\mathscr I)\Sigma (A+\mathscr I)$.  Furthermore, the condition $\mathrm{ker}(\Sigma_0)\subseteq\mathrm{ker}(\Sigma_1)$ (equivalently, $\overline{\mathrm{range}(\Sigma_0)}\supseteq\overline{\mathrm{range}(\Sigma_1)}$) is 
\begin{enumerate}
\item necessary and sufficient for the existence of the log map of $\Sigma_1$ at $\Sigma_0$, given by
\[
\log_{\Sigma_0}\Sigma_1= \topt{0}{1} - \mathscr I
=\Sigma_0^{-1/2}(\Sigma_0^{1/2}\Sigma_1\Sigma_0^{1/2})^{1/2}\Sigma_0^{-1/2} - \mathscr I,
\]
when it exists;
\item sufficient for the existence of a unique (unit speed) geodesic from $\Sigma_0$ to $\Sigma_1$ given by
\[
\Sigma_t
=[t\topt01 + (1-t)\mathscr I]\Sigma_0[t\topt01 + (1-t)\mathscr I]
=t^2\Sigma_1 + (1-t)^2\Sigma_0 + t(1-t)[\topt01\Sigma_0 + \Sigma_0\topt01],
\]
where again $\topt 01=\Sigma_0^{-1/2}(\Sigma_0^{1/2}\Sigma_1\Sigma_0^{1/2})^{1/2}\Sigma_0^{-1/2}$.
\end{enumerate}
Both points follow from the manifold properties of Wasserstein space discussed earlier in this subsection, by taking $\mu\equiv N(0,\Sigma_0)$ and $\nu\equiv N(0,\Sigma_1)$ and using Proposition~\ref{prop:existenceMaps} and the remarks on necessity thereafter.

\section{Topological Properties}\label{sec:topological}



The topological properties of the Wasserstein distance are well understood, as is the topic of weak convergence of Gaussian processes.  This knowledge can thus be used in order to understand the topology induced by the Procrustes distance.  Recall that a sequence of measures $\mu_n$ converges to $\mu$ in distribution (or narrowly)\footnote{This is often called weak convergence, but we will avoid this terminology in order to avoid confusion:  weak convergence of covariance operators is not equivalent to convergence in distribution of the corresponding measures.} if $\ownint {}{}f{\mu_n}\to \ownint {}{}{f}\mu$ for all continuous bounded $f:\mathcal H\to\mathbb R$.

\begin{proposition}[Procrustes Topology]
\label{prop:convEquiv}
Let $\{\Sigma_n\}_{n=1}^{\infty},\Sigma$ be covariance operators on $\mathcal{H}$.  The following are equivalent:
\begin{enumerate}
\item $ N(0,\Sigma_n)\stackrel{n\rightarrow\infty}{\longrightarrow} N(0,\Sigma)$ in distribution.
\item $\Pi(\Sigma_n,\Sigma)\stackrel{n\rightarrow\infty}{\longrightarrow} 0$.
\item $\hs\sqrt{\Sigma_n} - \sqrt{\Sigma}\hs_2 \stackrel{n\rightarrow\infty}{\longrightarrow} 0$.
\item $\hs\Sigma_n - \Sigma\hs_{1} \stackrel{n\rightarrow\infty}{\longrightarrow} 0$.
\end{enumerate}
\end{proposition}
In particular, sets of covariance opertators are pre-compact with respect to $\Pi$ if and only if the set of corresponding centred Gaussian measures with those covariances is uniformly tight.  We also remark that convergence in operator norm is \emph{not} sufficient for any of (1)-(4): to obtain a counterexample, take $\Sigma=0$ and let $\Sigma_n$ have $n$ eigenvalues equal to $1/n$ and all the others zero.
\begin{proof}
Write $\mu_n\equiv N(0,\Sigma_n)$ and $\mu\equiv N(0,\Sigma)$ for tidiness and recall that $\Pi(\Sigma_n,\Sigma)=W(\mu_n,\mu)$. For the implications (4)$\Longrightarrow$(1)$\iff$(3) see Examples~3.8.15 and 3.8.13(iii) in \citet{bogachev1998gaussian}.  By \cite[Theorem~3.8.11]{bogachev1998gaussian}, if (1) holds, then the measures $(\mu_n)$ have uniform exponential moments, and by Theorem~7.12 in \cite{villani2003topics} (1) and (2) are equivalent (see Corollary~\ref{cor:Gauss4mom} for a more elementary proof that does not involve Fernique's theorem).  To conclude it suffices to show that (2) yields (4).  Let $X\sim \mu$, $X_n\sim \mu_n$ (defined on the same probability space) such that $W^2(\mu_n,\mu)=\mathbb E\| X_n - X\|^2\to0$.  Notice that $\Sigma_n=\mathbb EX_n\otimes X_n$.  Invoking Jensen's inequality to
\[
\Sigma_n - \Sigma
=\mathbb EX_n\otimes (X_n - X) + \mathbb{E}(X_n - X)\otimes X
\]
yields (recall that $\hs f\otimes g\hs_1 = \|f\|\|g\|$;  see Lemma~\ref{lem:tensornorm} below)
\[
\hs \Sigma_n - \Sigma\hs_1
\le \mathbb E\hs X_n\otimes (X_n - X)\hs_1
+ \mathbb E \hs (X_n - X)\otimes X\hs_1
=\mathbb E\|X_n\|\|X_n - X\| + \|X_n - X\|\|X\|.
\]
When $n$ is sufficiently large $\mathbb E\|X_n\|^2 \le 1+\mathbb E\|X\|^2$ and then the right-hand side is
\[
\le \sqrt{\mathbb E\|X_n - X\|^2}\left(\sqrt {1+\mathbb E\|X\|^2} + \sqrt{\mathbb E\|X\|^2}\right)
=C(\Sigma)W(\mu_n,\mu),
\]
where $C(\Sigma)=\sqrt{1+\tr\Sigma} + \sqrt{\tr\Sigma}$, and this vanishes as $n\to\infty$.
\end{proof}
More is known about the topology of Wasserstein space;  for instance, the exponential and log maps given in Section~\ref{sec:geometry} are continuous, so $\W(\X)$ is homeomorphic to an infinite-dimensional convex subset of a Hilbert space $\mathcal L_2(\mu)$ (for any regular measure $\mu$); see the dissertation \citet[Lemmas~3.4.4 and 3.4.5]{zemel2017thesis} or the forthcoming book \citet{panaretos2018introduction}.

\section{Finite Rank Approximations}\label{sec:projections} 

\citet{pigoli2014distances} considered the validity of approximating the Procrustes distance $\Pi$ between two infinite-dimensional operators by the distance between finite-dimensional projections thereof (in the sense of convergence of the latter to the former). Though this validity can be obtained from Proposition \ref{prop:convEquiv}, use of the Wasserstein interpretation of $\Pi$ provides a straightforward calculation of the projection error and an elementary proof of convergence under projections forming an approximate identity in $\mathcal{H}$ (whether the projections are finite dimensional or not).  If the projections are indeed finite dimensional, one can furthermore establish a stronger form of validity: uniform convergence over compacta.  

Let $\mu\in \mathcal W(\X)$ with covariance $\Sigma$ and $\mathscr P$ be a projection operator ($\mathscr P^*=\mathscr P=\mathscr P^2$).  Then $\mathscr P$ is an optimal map from $\mu$ to $\mathscr P\#\mu$ and so
\[
W^2(\mu,\mathscr P\#\mu)
=\ownint {\X}{}{\|x-\mathscr Px\|^2}{\mu(x)}
=\tr\left\{(\mathscr I-\mathscr P)\Sigma(\mathscr I-\mathscr P)\right\}
=\tr\left\{(\mathscr I-\mathscr P)\Sigma\right\}.
\]
This is true regardless of $\mu$ being Gaussian, but taking $\mu$ to be $N(0,\Sigma)$, in particular, yields the explicit error
$$\Pi^2(\Sigma,\mathscr{P}\Sigma\mathscr{P}) =\tr\left\{(\mathscr I-\mathscr P)\Sigma\right\},$$
where $\mathscr{P}\Sigma\mathscr{P}$ is the projection of $\Sigma$ onto the range of $\mathscr{P}$. This indeed converges to zero when $\mathscr{P}_n$ is an approximate identity, in the sense of $\mathscr{P}_n$ converging \emph{strongly} to the identity: a sequence of operators $T_n$ converges to $T$ \emph{strongly} if $T_nx\to Tx$ for all $x\in \X$ (\citet[p.\ 198]{stein2009real}).\footnote{This is much weaker than convergence in operator norm, but stronger than requiring that $\innprod {T_nx}y\to\innprod {Tx}y$ for all $x,y\in H$, which is called \emph{weak} convergence of $T_n$ to $T$.\label{footnote:strongly}}

\begin{lemma}\label{lem:strong_convergence}
Let $\mathscr P_n$ be a sequence of projections that converges strongly to the identity.  Then $\mathscr P_n\#\mu\to \mu$ in $\mathcal W(\X)$ for any $\mu\in\mathcal W(\X)$, and consequently $\Pi(\Sigma,\mathscr{P}_n\Sigma\mathscr{P}_n)\rightarrow 0$.
\end{lemma}

The setting considered in \citet{pigoli2014distances} is indeed a special case of Lemma \ref{lem:strong_convergence}:  let $\{e_k\}_{k\geq 1}$ be an orthonormal basis of $\X$ and define $\mathscr P_n=\sum_{j=1}^{n}e_j\otimes e_j$ as the projection onto the span of $\{e_1,\dots,e_n\}$.  Then $\mathscr P_n$ converges strongly to the identity as $n\rightarrow\infty$.

\begin{proof}[Proof of Lemma \ref{lem:strong_convergence}]
Since $\mathscr P_nx\to x$ for all $x$ and $\|\mathscr P_nx\|\le\|x\|$, the result that $\mathscr P_n\#\mu\to \mu$ in $\mathcal W(\X)$  follows from the dominated convergence theorem. Taking $\mu\equiv N(0,\Sigma)$ then completes the proof.
\end{proof}
When focussing on finite dimensional projections, a stronger statement is possible, if one considers compact sets:
\begin{proposition}\label{prop:unifConv}
Let $\{e_k\}_{k\geq 1}$ be an orthonormal basis of $\mathcal{H}$ and $\mathscr P_n=\sum_{j=1}^{n}e_j\otimes e_j$ be the projection on the span of $\{e_1,\dots,e_n\}$.  Let 
$\mathcal B$ be a collection of positive bounded operators satisfying
\begin{equation}\label{equi_small_tails}
\sup_{\Sigma\in \mathcal B}
\sum_{j=n+1}^\infty
\innprod {\Sigma e_j}{e_j}
\to0
,
\qquad\qquad
\textrm{as }
n\to\infty.
\end{equation}
Then,
\[
\sup_{\Sigma_1,\Sigma_2\in\mathcal B}
|\Pi(\mathscr P_n\Sigma_1\mathscr P_n,\mathscr P_n\Sigma_2\mathscr{P}_n) - \Pi(\Sigma_1,\Sigma_2)|
\to0
,
\qquad
n\to\infty.
\]
\end{proposition}
\begin{proof}
Let $\mathcal K\subset \mathcal W(\X)$ be a collection of measures with $m(\mu)\in A$ and 
$\Sigma(\mu)\in\mathcal B$ for all $\mu\in\mathcal K$. It suffices to show that $W(\mu,\mathscr P_n\#\mu)\to0$ uniformly and indeed
\[
W^2(\mu,\mathscr P_n\#\mu)
=
\tr(\mathscr I-\mathscr P_n)\Sigma(\mu)
=
\sum_{j=n+1}^\infty \innprod {\Sigma(\mu) e_j}{e_j}
\]
vanishes uniformly as $n\to\infty$.
\end{proof}
The collection $\mathcal B$ of covariances of a tight set of centred Gaussian measures satisfies the tail condition \eqref{equi_small_tails} with respect to \emph{any} orthonormal basis $\{e_k\}_{k\ge1}$ of $\X$ (\cite[Example 3.8.13(iv)]{bogachev1998gaussian}).  As per Proposition~\ref{prop:convEquiv}, the tightness condition admits three alternative equivalent formulations in purely operator theory terms.  The first is that $\mathcal B$ be compact with respect to the distance $\Pi$.  The second is that $\mathcal B$ be of the form $\mathcal{B}=\{A^2: A\in \mathcal{A}\}$ for $\mathcal{A}$ a compact set of positive Hilbert--Schmidt operators.  The third is that $\mathcal B$ be compact with respect to the trace norm.


\section{Existence and Uniqueness of Fr\'echet Means}\label{sec:existUnique}
The most basic statistical task in a general metric space is that of obtaining a notion of average.  If $\Sigma_1,\dots,\Sigma_N$ are covariance operators, their mean can be modelled as a Fr\'echet mean (\citet{frechet1948elements}) with respect to the Procrustes metric (equivalenly, a Wasserstein barycentre of corresponding centred Gaussian measures), defined as the minimiser of the Fr\'echet functional
\[
F(\Sigma)
=\frac1{2N}\sum_{i=1}^N
\Pi^2(\Sigma,\Sigma_i)
=\frac1{2N}\sum_{i=1}^N
W^2(N(0,\Sigma),N(0,\Sigma_i)).
\]
One can also can consider the Fr\'echet mean of a random covariance operator $\mathscr A$ as the minimiser of $\Sigma\mapsto F(\Sigma)=\frac12\mathbb E\Pi^2(\Sigma,\mathscr A)$;  the empirical measure can be recovered from this when $\mathscr A$ has the uniform distribution on the finite set $\{\Sigma_1,\dots,\Sigma_N\}$. See Section \ref{sec:generative} for a more thorough discussion of the population case. 
  The Fr\'echet mean of arbitrary measures in $\mathcal W(\mathcal H)$ can be defined in the same way.  Unlike the linear mean, existence and uniqueness of Fr\'echet means in general metric space is a rather delicate matter (see, e.g., \citet{bhattacharya2003large,bhattacharya2005large} and \citet{karcher1977riemannian}).  In the particular case of the Wasserstein space, however, existence and uniqueness can be established under rather mild assumptions.  For the finite-dimensional case, such conditions were studied by \citet{agueh2011barycenters}.  In particular, it is known that the Fr\'echet mean of Gaussian measures is a Gaussian measure, and so there is no ambiguity as to whether we minimise $F$ over Gaussian measures or arbitrary measures.

\citet{pigoli2014distances} et al also considered the Fr\'echet mean with respect to $\Pi$, but working with their formulation $\Pi(\Sigma_1,\Sigma_2)=
\inf_{U:\,U^* U=\mathscr{I}}\hs\Sigma_1^{1/2}-U\Sigma_2^{1/2}\hs_2$ made it difficult to deal with existence and uniqueness.  We now show how this can be done easily using the Wasserstein interpretation.  We begin with existence, which holds for a general collection of measures.  The proof relies upon the notion of multicouplings.
\begin{definition}[multicouplings]
Let $\mu_1,\dots,\mu_N\in\mathcal W(\mathcal H)$.  A multicoupling of $(\mu_1,\dots,\mu_N)$ is a Borel measure on $\mathcal H^N$ with marginals $\mu_1,\dots,\mu_N$.
\end{definition}
An optimal multicoupling of $\mu_1,\dots,\mu_N$ is a multicoupling $\pi$ that minimises
\[
G(\pi) 
=\frac 1 {2N^2} \ownint {\mathcal H^N}{}{\sum_{i<j} \|x_i - x_j\|^2}{\pi(x_1,\dots,x_N)}
=\ownint {\mathcal H^N}{}{\frac 1{2N} \sum_{i=1}^N \|x_i - \overline x\|^2}{\pi(x)}.
\]
We shall discuss the probabilistic interpretation of multicoupling in more detail in Section~\ref{sec:pca};  at this stage we merely use it as a tool for deriving analytical properties of Fr\'echet means.  When $N=2$, multicouplings are simply couplings and finding an optimal multicoupling is the optimal transport problem.  On $\mathbb R^d$, multicouplings were studied by \citet{gangbo1998optimal} (also see \citet{zemel2017fr}).  In analogy with the optimal transport problem, an optimal multicoupling always exists, and if $\mu_1$ is regular an optimal multicoupling takes the form $(\mathscr I,S_2,\dots,S_N)\#\mu_1$ for some functions $S_i:\mathbb R^d\to\mathbb R^d$, where
\[
(\mathscr I,S_2,\dots,S_N)\#\mu_1(B_1\times\hdots\times B_N)
=\mu_1(\{x\in B_1:S_2(x)\in B_2,\dots,S_N(x)\in B_N\})
=\mu_1\left(\bigcap_{i=1}^NS_i^{-1}(B_i)\right)
\]
for any Borel-rectangle $B_1\times\hdots\times B_N$, and $S_1=\mathscr I$. The relationship between multicouplings and Fr\'echet mean becomes clear in the following lemma.  It is a slight refinement of Proposition~4.2 in \citet{agueh2011barycenters}, and we provide a proof for completeness.

\begin{lemma}[Fr\'echet means and multicouplings]\label{lem:FrechetMulti}
Let $\mu^1,\dots,\mu^N\in \mathcal W$.  Then $\mu$ is a Fr\'echet mean of $(\mu^1,\dots,\mu^N)$ if and only if there exists a multicoupling $\pi\in \mathcal W _2(\mathcal H^N)$ of $(\mu^1,\dots,\mu^N)$ such that
\[
\mu
=
M_N \# \pi
,
\qquad
M_N:\mathcal H^N \to \mathcal H
,
\qquad
M_N(x_1,\dots,x_N)
=\overline x
=
\frac 1N \sum_{i=1}^N x_i.
\]
\end{lemma}
\begin{proof}
Let $\pi$ be an arbitrary multicoupling of $(\mu^1,\dots,\mu^N)$ and set $\mu=M_N\#\pi$.  Then $(x\mapsto x_i,M_N)\#\pi$ is a coupling of $\mu^i$ and $\mu$, and therefore
\[
\ownint {\mathcal H^N}{}{\|x_i - M_N(x)\|^2}{\pi(x)}
\ge
W^2(\mu,\mu_i).
\]
Summation over $i$ gives $F(\mu)\le G(\pi)$ and so $\inf F\le \inf G$.

For the other inequality, let $\mu\in\mathcal W$ be arbitrary.  For each $i$ let $\pi^i$ be an optimal coupling between $\mu$ and $\mu^i$.  Invoking the gluing lemma (Ambrosio \& Gigli \cite[Lemma 2.1]{ambrosio2013user}), we may glue all $\pi^i$'s using their common marginal $\mu$.  This procedure constructs a measure $\eta$ on $\mathcal H^{N+1}$ with marginals $\mu_1,\dots,\mu_N,\mu$ and its relevant projection $\pi$ is then a multicoupling of $\mu_1,\dots,\mu_N$.

Since $\mathcal H$ is a Hilbert space, the minimiser of $y\mapsto \sum \|x_i - y\|^2$ is $y=M_N(x)$.  Thus
\[
F(\mu)
=
\frac 1{2N} \ownint {\mathcal H^{N+1}}{}{\sum_{i=1}^N \|x_i - y\|^2} {\eta(x,y)}
\ge
\frac 1{2N} \ownint {\mathcal H^{N+1}}{}{\sum_{i=1}^N \|x_i - M_N(x)\|^2} {\eta(x,y)}
=
G(\pi).
\]
In particular, $\inf F\ge \inf G$ and combining this with the established converse inequality we see that $\inf F=\inf G$.  Observe also that the last displayed inequality holds as equality if and only if $y=M_N(x)$ $\eta$-almost surely, in which case $\mu=M_N\#\pi$.  Therefore if $\mu$ does not equal $M_N\#\pi$, then $F(\mu)>G(\pi)\ge F(M_N\#\pi)$, and $\mu$ cannot be optimal.  Finally, if $\pi$ is optimal, then
\[
F(M_N\#\pi)
\le
G(\pi)
=\inf  G
=\inf F
\]
establishing optimality of $\mu=M_N\#\pi$ and completing the proof.
\end{proof}
\begin{corollary}[Fr\'echet means and moments]\label{cor:FrechetMom}
Any finite collection of measures $\mu^1,\dots,\mu^N\in \mathcal W(\mathcal H)$ admits a Fr\'echet mean $\mu$, for all $p\ge1$
\[
\ownint {\mathcal H}{}{\|x\|^p}{\mu(x)}
\le
\frac 1N  \sum_{i=1}^N
\ownint {\mathcal H}{}{\|x\|^p}{\mu^i(x)},
\]
and when $p>1$ equality holds if and only if $\mu^1=\dots=\mu^N$.  In particular, any collection $\Sigma^1,\dots,\Sigma^N$ of covariance operators admits a Fr\'echet mean $\overline\Sigma$ with respect to the Procrustes distance $\Pi$, and $\tr\overline\Sigma\le N^{-1}\sum_{i=1}^N\tr\Sigma^i$.
\end{corollary}
\begin{proof}
Let $\pi$ be a multicoupling of $\mu^1,\dots,\mu^N$ such that $\mu=M_N\#\pi$ (Lemma~\ref{lem:FrechetMulti}).  Then
\[
\ownint {\mathcal H}{}{\|x\|^p}{\mu(x)}
=\ownint {\mathcal H^N}{}{\left\|\frac 1N \sum_{i=1}^N x_i\right\|^p}{\pi(x)}
\le
\frac 1N \sum_{i=1}^N \ownint {\mathcal H^N}{}{\|x_i\|^p}{\pi(x)}
=\frac 1N \sum_{i=1}^N \ownint {\mathcal H}{}{\|x\|^p}{\mu^i(x)}.
\]
The statement about equality follows from strict convexity of $x\mapsto\|x\|^p$ if $p>1$.
\end{proof}
We next turn to uniqueness of Fr\'echet means.  The proof follows from strict convexity of the Fr\'echet functional $F$ that manifests as soon as enough non-degeneracy is present.  We remark first that if $\mu$ is a Gaussian measure with covariance $\Sigma$, then $\mu$ is regular if and only if $\Sigma$ is injective, and when this holds, for any $\nu\in \W(\mathcal H)$ (Gaussian or not) the optimal map $\topt\mu\nu$ exists.

\begin{proposition}
\label{prop:uniqueFrechet}
Let $\mu_1,\dots,\mu_N\in \W(\X)$ and assume that $\mu_1$ is regular.  Then the Fr\'echet functional is strictly convex, and the Fr\'echet mean of $\mu_1,\dots,\mu_N$ is unique.  In particular, the Fr\'echet mean of a collection of covariance operators is unique if at least one of the operators is injective.
\end{proposition}
Uniqueness in fact holds at the population level as well:  the condition is that the random covariance operator be injective with positive probability.  On $\mathbb R^d$ this was observed by \citet{bigot2012characterization} in a parametric setting, and extended to the nonparametric setting by \citet{zemel2017fr};  the analytical idea dates back to \citet{alvarez2011uniqueness}.

\begin{proof}
We first establish weak convexity of the squared Wasserstein distance. Let $\nu_1,\nu_2,\mu\in\W(\X)$ and let $\pi_i$ be an optimal coupling of $\nu_i$ and $\mu$.  For any $t\in(0,1)$ the linear interpolant $t\pi_1+(1-t)\pi_2$ is a coupling of $t\nu_1 + (1-t)\nu_2$ and $\mu$.  This yields the weak convexity
\begin{equation}\label{eq:FrechetFunctionalConvex}
W^2(t\nu_1+(1-t)\nu_2,\mu)
\le \ownint {\X^2}{}{\|x - y\|^2}{[t\pi_1+(1-t)\pi_2](x,y)}
=tW^2(\nu_1,\mu)
+(1-t)W^2(\nu_2,\mu).
\end{equation}
Now if $\mu$ is regular, then both couplings $\pi_i$ are induced by maps $T_i=\topt\mu{\nu_i}$.  If $\nu_1\ne\nu_2$, then $t\pi_1+(1-t)\pi_2$ is not induced from a map, and consequently cannot be the optimal coupling of $t\nu_1+(1-t)\nu_2$ and $\mu$.  Thus the inequality above is strict and $W^2(\cdot,\mu)$ is strictly convex.  The proposition now follows upon noticing that the Fr\'echet functional is a sum of $N$ squared Wasserstein distances that are all convex, one of them strictly.
\end{proof}
For statistical purposes existence and uniqueness are not sufficient, and one needs to find a constructive way to evaluate the Fr\'echet mean of a given collection of covariance operators. \citet{pigoli2014distances} propose using the classical generalised Procrustes algorithm.  The Wasserstein formalism gives rise to another algorithm that can be interpreted as steepest descent in Wasserstein space, while still admitting a Procrustean interpretation (solving successive pairwise \emph{transport} rather than \emph{alignment} problems).  We will elaborate on these algorithms in Section~\ref{sec:algorithms}.  In practice, implementing these algorithms will require finite-dimensional versions of the operators.  This raises the question of stability of the Fr\'echet mean under projections, which is the topic of the next section.



\section{Stability of Fr\'echet Means}\label{sec:stability}





When analysing functional data, one seldom has access to the genuinely infinite-dimensional objects (see, e.g., \citet{hsing2015theoretical}, \citet{yao2005functional,yao2005functional2}, \citet{descary2016functional}).  In practice, the observed curves are discretised at some level and the data at hand represent finite-dimensional approximations, potentially featuring some additional level of smoothing.   It is therefore important to establish some amount of continuity of any inferential procedure with respect to progressively finer such approximations. In the present context, it is important to verify that the Fr\'echet mean remains stable as the discretisation becomes finer and finer, and as smoothing parameters decay.  Stability of the Procrustes distance itself (as, e.g., in Section~\ref{sec:projections}), does not immediately yield stability of the Fr\'echet means -- the latter amounts to argmin theorems (\citet{van1996weak}), whose validity requires further assumptions.
Our understanding of the topology of the Wasserstein space, however, allows to deduce this stability of the Fr\'echet means.  We note that the question of convergence of Fr\'echet means on locally compact spaces was studied by \citet{gouic2016existence}, but their results cannot be applied in our setup, since $\mathcal H$ is not locally compact.
\begin{theorem}
[Fr\'echet means and projections]
\label{thm:convNfixed}
Let $\Sigma^1,\dots,\Sigma^N$ be covariance operators with $\Sigma^1$ injective, and let $\{\Sigma^i_k:i\leq N,\,k\geq 1\}$ be sequences such that $\Sigma^i_k\stackrel{k\rightarrow\infty}{\longrightarrow}\Sigma^i$ in trace norm (equivalently, in Procrustes distance).  Then (any choice of) the Fr\'echet mean of $\Sigma^1_k,\dots,\Sigma^N_k$ converges in trace norm to that of $\Sigma^1,\dots,\Sigma^N$.
\end{theorem}
Of course, the result can be phrased in terms of Gaussian measures $\{N(0,\Sigma^i):i\leq N\}$ and sequences $\{N(0,\Sigma^i_k):i\leq N,\,k\geq 1\}$, and their Wasserstein barycentres. But, in fact, the result holds far more generally in the space $\mathcal{W}(\X)$.  Let $\mu^i_k\to \mu^i$ in Wasserstein distance.  If for some $\epsilon>0$ and all $i=1,\dots,N$
\[
\sup_{k}\ownint{\mathcal H}{}{\|x\|^{2+\epsilon}}{\mu^i_k(x)}
<\infty
\]
then the statement of the theorem holds.  The proof follows by modifying the exponent in the definition of $R^i$ in step 2 to $2+\epsilon$ and noticing that the resulting moment bound for $\overline\mu_k$ supplemented by their tightness yields convergence in Wasserstein distance;  see the discussion following \cite[Equation~(5.3)]{zemel2017fr} (and replace 3 by $2+\epsilon$).

\begin{proof}[Proof of Theorem~\ref{thm:convNfixed}]
Denote the corresponding measures Gaussian measures $\{N(0,\Sigma^i):i\leq N\}$ and sequences $\{N(0,\Sigma^i_k):i\leq N,\,k\geq 1\}$ by $\{\mu^i:i\leq N\}$ and $\{\mu^i_k:i\leq N,\,k\geq 1\}$ and let $\overline\mu_k$ denote any Fr\'echet mean of $\mu^1_k,\dots,\mu^N_k$.

\textbf{Step 1:}  tightness of $(\overline\mu_k)$.  The entire collection $\mathcal K=\{\mu^i_k\}$ is tight, since all the sequences converge in distribution (Proposition~\ref{prop:convEquiv}).  For any $\epsilon>0$ there exists a compact $K_\epsilon\subset \mathcal H$ such that $\mu(K_\epsilon)\ge 1-\epsilon/N$ for all $\mu\in\mathcal K$.  Replacing $K_\epsilon$ by its closed convex hull (Lemma~\ref{lem:compactHull}), we may assume it to be convex as well.

Let $\pi_k$ be any multicoupling of $(\mu^1_k,\dots,\mu^N_k)$.  Then the marginal constraints of $\pi_k$ imply that $\pi_k(K_\epsilon^N)\ge 1-\epsilon$.  By Lemma~\ref{lem:FrechetMulti}, $\overline\mu_k$ must take the form $M_N\#\pi_k$ for some multicoupling $\pi_k$.  Convexity of $K_\epsilon$ implies that $M_N^{-1}(K_\epsilon)\supseteq K_\epsilon^N$, and so
\[
\overline\mu_k(K_\epsilon)
=
\pi_k(M_N^{-1}(K_\epsilon))
\ge
\pi_k(K_\epsilon^N)
\ge 1-\epsilon.
\]
With tightness of $(\overline\mu_k)$ established, we may now assume that (up to subsequences) $\overline\mu_k$ converge in distribution to a limit $\overline\mu$.  Since $\overline\mu_k$ are Gaussian, they also converge in Wasserstein distance by Proposition~\ref{prop:convEquiv}.

\textbf{Step 2:}  a moment bound for $\overline\mu_k$.  Let $R^i=\ownint {\mathcal H}{}{\|x\|^2}{\mu^i(x)}$ denote the second moment of $\mu^i$.  Since the second moments can be interpreted as a (squared) Wasserstein distance to the Dirac mass at 0 (or by Theorem~7.12 in \cite{villani2003topics}), the second moment of $\mu^i_k$ converges to $R^i$ and so for $k$ large it is smaller than $R^i+1$.  By Corollary~\ref{cor:FrechetMom}, for $k$ large
\[
\ownint {\mathcal H}{}{\|x\|^2}{\overline\mu_k(x)}
\le
\frac 1N \sum_{i=1}^N R^i + 1
\le
\max(R^1,\dots,R^N) + 1
:= R + 1.
\]

\textbf{Step 3:}  the limit $\overline\mu$ is a Fr\'echet mean of $(\mu^i)$.  By the moment bound above, the Fr\'echet means $\overline\mu_k$ can be found (for $k$ large) in the Wasserstein ball
\[
B
=
\{\mu\in \W:
W^2(\mu,\delta_0)
\le R+1
\},
\]
with $\delta_0$ a Dirac measure at the origin.  If $\mu,\nu\in B$ then, since $\mu^i_k\in B$ for $k$ large,
\[
|
F_k(\mu)
-
F_k(\nu)
|
\le
\frac 1{2N}
\sum_{i=1}^N
[W(\mu,\mu^i_k) + W(\nu,\mu^i_k)]
W(\mu,\nu)
\le 2 \sqrt {R+1}\ 
W(\mu,\nu).
\]
In other words, all the $F_k$'s are uniformly Lipschitz on $B$.  Suppose now that $\overline\mu_k\to\overline\mu$ in $\W$.  Let $\mu\in B$, $\epsilon>0$ and $k_0$ such that $W(\overline\mu_k,\overline\mu)<\epsilon/(2\sqrt{R+1})$ for all $k\ge k_0$.  Since $F_k\to F$ pointwise we may assume that $|F(\mu) - F_k(\mu)| < \epsilon$ when $k\ge k_0$ and the same holds for $\mu=\overline\mu$.  Then for all $k\ge k_0$
\[
\epsilon + F(\mu)
\ge F_k(\mu)
\ge F_k(\overline\mu_k)
\ge F_k(\overline\mu) - \epsilon
\ge F(\overline\mu) - 2\epsilon.
\]
Since $\epsilon$ is arbitrary we see that $\overline\mu$ minimises $F$ over $B$ and hence over the entire Wasserstein space $\W(\X)$.

\textbf{Step 4:}  conclusion.  We have shown or assumed that
\begin{itemize}
\item the sequence $(\overline\mu_k)$ is precompact in $\W(\X)$;
\item each of its limits is a minimiser of $F$;
\item there is only one minimiser of $F$.
\end{itemize}
The combination of these three facts implies that $\overline\mu_k$ must converge to the minimiser of $F$.
\end{proof}

\section{Computation of Fr\'echet Means: Procrustes Algorithms and Gradient Descent}
\label{sec:algorithms}




 Fr\'echet means rarely admit closed-form expressions, and the Procrustes space of covariances on $\mathcal{H}$ (equivalently, the Wasserstein space of Gaussian measures on $\mathcal{H}$) is no exception (but see Section~\ref{sec:fixedpoint} for the issue of characterisation).  In order to compute the Fr\'echet mean in practice, one needs to resort to numerical schemes at some level, and such schemes would need to be applied to finite-dimensional versions of the covariances, resulting from the necessarily discrete nature of observation and/or smoothing.

Let $\Sigma_1,\dots,\Sigma_N$ be covariance operators, of which one seeks to find a Fr\'echet mean $\overline \Sigma$.  \citet{pigoli2014distances} suggested an iterative procedure, motivated by generalised Procrustes analysis (\citet{gower1975generalized}; \citet{dryden1998statistical}), for finding $\mathscr L=\overline \Sigma^{1/2}$, that we summarise as follows.  The initial point $\mathscr L^0$ is the average of $L_i^0=L_i=\Sigma_i^{1/2}$.  At step $k$, one computes, for each $i$, the unitary operator $R_i$ that minimises $\hs \mathscr L^{k-1} - \mathscr L_i^{k-1}R_i\hs$, and then sets $\mathscr L_i^k=\mathscr L_i^{k-1}R_i$.  After this, one defines $\mathscr L^k$ as the average of $\{\mathscr L_1^k,\dots,\mathscr L_N^k\}$ and repeats until convergence.  The advantage of this algorithm is that it only involves successively matching pairs of operators (minimising $\hs \mathscr L^{k-1} - \mathscr L_i^{k-1}R_i\hs$), for which there is an explicit solution in terms of the SVD of the product of the operators in question.  \citet{pigoli2014distances} report good empirical performance of this algorithm (and some of its variants) on discretised versions of the operators {if the initial point is chosen as $n^{-1}\sum_{i=1}^n\Sigma_i^{1/2}$}, and conjecture that an infinite-dimensional implementation of their algorithm would also exhibits favourable performance. {It is not clear whether the algorithm converges, though, when the operators $\{\Sigma_1,\dots,\Sigma_N\}$ do not commute}.

The great advantage of the procedure of \citet{pigoli2014distances} is precisely that it only involves successive averaging of solutions of pairwise matching problems until convergence.  The Wasserstein formalism allows one to construct an alternative algorithm, that is also similar in spirit to generalised Procrustes analysis: instead of averaging pairwise SVD matchings, one averages pairwise optimal transport maps. This algorithm was proposed independently and concurrently by \citet{zemel2017fr} and \citet{alvarez2016fixed} in a finite dimensional setting.  Though it can be applied to any finite collection of measures in Wasserstein space, we shall outline it here in the setup of covariance operators only (i.e., for collections of centred Gaussian measures).  Let $\Sigma^0$ be an injective initial point and suppose that the current iterate at step $k$ is $\Sigma^k$.  For each $i$ compute the optimal maps from $\Sigma^k$ to each of the prescribed operators $\Sigma_i$, namely $\topt{\Sigma^k}{\Sigma_i}=(\Sigma^k)^{-1/2}[(\Sigma^k)^{1/2}\Sigma_i(\Sigma^k)^{1/2}]^{1/2}(\Sigma^k)^{-1/2}$.  Define their average $T_k=N^{-1}\sum_{i=1}^N\topt{\Sigma^k}{\Sigma_i}$, a positive (possibly unbounded) operator, and then set the next iterate to $\Sigma^{k+1}=T_k\Sigma^kT_k$. 

In terms of the manifold geometry of covariances under the Procrustes metric (see Section \ref{sec:geometry}), the algorithm starts with an initial guess of the Fr\'echet mean; it then lifts all observations to the tangent space at that initial guess via the log map, and averages linearly on the tangent space; this linear average is then retracted onto the manifold via the exponential map, providing the next guess, and iterates.

In finite dimensions\footnote{i.e. when the operators $\Sigma_1,\dots,\Sigma_N$ are of finite rank, which will always be the case in practice, as explained in the first paragraph of this Section}, the Wasserstein-inspired algorithm is shown \cite{zemel2017fr,alvarez2016fixed} to converge to the unique Fr\'echet mean $\overline\Sigma$ of $\Sigma_1,\dots,\Sigma_N$ provided one of them is injective, and this independently of the initial point.  Moreover,  \citet{alvarez2016fixed} show $\tr\Sigma^k$ to be increasing in $k$, and \citet{zemel2017fr} show that the optimal maps  $\topt{\Sigma^k}{\Sigma_i}$ converge uniformly over compacta to $\topt{\overline\Sigma}{\Sigma_i}$ as $k\to\infty$. In fact, in finite dimensions, \citet{zemel2017fr} demonstrate that this algorithm is classical steepest descent in Wasserstein space (in our setting, it is steepest descent in the space of covariances endowed with the Procrustes metric $\Pi$).  

{Compared to the procedure of \citet{pigoli2014distances}, the Wassestein-inspired algorithm appears to be numerically more stable. For example, we observed through simulations that it is less sensitive to the initial point.}
Finally, it is worth mentioning that when the covariance operators commute, either algorithm converges to the Fr\'echet mean after a single iteration.\footnote{In the Wasserstein-inspired algorithm this requires to start from any positive linear combination of the operators themselves, or positive powers thereof.}

\section{Tangent Space PCA, Optimal Multicoupling, and Amplitude vs Phase Variation}\label{sec:pca}
Once a Fr\'echet mean of a given sample of covariance operators is found, the second order statistical analysis is to understand the variation of the sample around this mean.  The optimal value of the Fr\'echet functional gives a coarse measure of variance (as a sum of squared distances of the observation from their mean), but it is desirable to find a parsimonious representation for the main sources/paths of variation in the sample, analogous to Principal Component Analysis (PCA) in Euclidean spaces \citep{jolliffe2002principal} and \emph{functional} versions thereof in Hilbert spaces \citep{panaretos2013cramer}.

One way of carrying out PCA in non-Euclidean spaces is by working on the tangent space (\citet{huckemann2010intrinsic}, \citet{fletcher2004principal} and \citet{dryden2009non}).  In the setup of covariance operators with the Procrustes distance $\Pi$, this can be done in closed form.  Using the log map at $\overline \Sigma$, one lifts the data $\Sigma_1,\dots,\Sigma_n$ to the points $\log_{\overline \Sigma}(\Sigma_i)=\topt{\overline \Sigma}{\Sigma_i} - \mathscr I=\overline\Sigma^{-1/2}[\overline\Sigma^{1/2}\Sigma_i\overline\Sigma^{1/2}]^{1/2}\overline\Sigma^{-1/2} - \mathscr I$ in the tangent space at the Fr\'echet mean, $\Tan_{\overline \Sigma}$ (see Section~\ref{sec:geometry}).  One can then carry out linear PCA of the data at the level of the tangent space.  The resulting components, orthogonal segments in $\Tan_{\overline \Sigma}$, can then be retracted to the space of covariance operators by means of the exponential map $\exp_{\overline \Sigma}$.  These would give principal geodesics that explain variation in the data;  retracting linear combinations of the principal components would result in principal submanifolds.

Since the tangent space is a linear approximation to the manifold, the success of tangent space PCA in explaining the variability of the sample around its mean depends on the quality of the approximation.  In finite dimensions, the typical difficulty comes from the cut locus of the manifold;  the log map is not defined on the entire manifold, and one often needs to assume that the spread of the observations around the mean is not too large.  In the Wasserstein space, this is actually not a problem, since the exponential map is surjective under sufficient injectivity (see Section~\ref{sec:geometry}).  The difficulty here is of a rather different nature, and amounts precisely to verification that required injectivity takes place.  The issue is that the log map $\log_{\overline \Sigma}(\Sigma_i)$ at $\overline \Sigma$ is well-defined \emph{if and only if} $\mathrm{ker}(\overline\Sigma)\subseteq \mathrm{ker}(\Sigma_i)$ (Proposition~\ref{prop:existenceMaps}, and discussion thereafter; equivalently, one requires $\overline{\mathrm{range}(\overline\Sigma)}\supseteq \overline{\mathrm{range}(\Sigma_i)}$).  In finite dimensions, we know that if one $\Sigma_i$ is injective (nonsingular), then so is the Fr\'echet mean, so the log map is well defined.  We conjecture that the same result holds in infinite dimensions, and leave this question for future work (see Section~\ref{sec:future}).

It is important to remark that in practice, the tangent space PCA will only be employed at a discretised level (and thus in finite dimensions), where it is indeed feasible and guaranteed to make sense.  The question is whether the procedure remains stable as the dimensionality of the discretisation grows to infinity.  The stability of the Wasserstein distance (Section~\ref{sec:projections}) and the Fr\'echet mean (Theorem~\ref{thm:convNfixed}) suggests that this should be so, but a rigorous proof amounts to establishing injectivity as in the preceding paragraph.

Tangent space PCA pertains to the collection of observations $\Sigma_1,\dots,\Sigma_n$ as a whole, and is consequently intimately related to multicoupling of the corresponding measures, admitting a further elegant interpretation.  Recall from Section~\ref{sec:existUnique} that the problem of optimal multicouplings consists of minimising the functional
\[
G(\pi)
=\frac 1 {2n^2} \ownint {\mathcal H^n}{}{\sum_{i<j} \|x_i - x_j\|^2}{\pi(x_1,\dots,x_n)}
\]
over all Borel measures $\pi$ on $\X^n$ having $\mu_1,\dots,\mu_n$ as marginals.  In other words, we seek to multicouple (the centred Gaussian measures corresponding to) $\Sigma_1,\dots,\Sigma_n$ as closely as possible, that is, in such a way that the sum of pairwise squared distances between the covariances is minimal.  The probabilistic interpretation is that one is given random variables $X_i\sim \mu_i$ and seeks to construct a random vector $(Y_1,\dots,Y_n)$ on $\X^n$ such that $Y_i\stackrel{d}{=}X_i$ marginally, and such that
\[
\mathbb E\sum_{i<j} \|Y_i - Y_j\|
\qquad
\textrm{is minimal}.
\]
Intuitively, one wishes to construct a vector on $\mathcal{H}^n$, whose coordinates are maximally correlated, subject to having prescribed marginal distibutions.  In Lemma~\ref{lem:FrechetMulti} we have seen that an optimal multicoupling yields the Fr\'echet mean.  However, as observed in \citet{zemel2017fr}, the proof actually allows to go in the other direction, and \emph{deduce an optimal multicoupling from the Fr\'echet mean}.   In the probabilistic terminology, we can write this down formally as follows:
\begin{lemma}\label{lem:frechetgivesmulti}
Let $\Sigma_1,\dots,\Sigma_n$ with injective Fr\'echet mean $\overline \Sigma$.  Let $Z\sim N(0,\overline\Sigma)$ and define a random Gaussian vector on $\X^n$ by
\[
(Y_1,\dots,Y_n)
,\qquad
Y_i = \topt{\overline\Sigma}{\Sigma_i}(Z)
=\overline\Sigma^{-1/2}[\overline\Sigma^{1/2}\Sigma_i\overline\Sigma^{1/2}]^{1/2}\overline\Sigma^{-1/2}Z
,\qquad i=1,\dots,n.
\]
Then, the joint law of $(Y_1,\dots,Y_n)$ is an optimal multicoupling of $\Sigma_1,\dots,\Sigma_n$.
\end{lemma}
We can reformulate Lemma~\ref{lem:frechetgivesmulti} as an optimisation problem on the space of covariance operators on the tensor product Hilbert space $\X^n$.  Define the coordinate projections $p_i:\X^n\to \X$ by $\pi_i(h_1,\dots,h_n)=h_i$.  The problem is to construct a covariance operator $\mathbf \Sigma$ on $\X^n$ that, under the marginal constraints $p_i\mathbf \Sigma p_i^*=\Sigma_i$, maximises
\[
\tr\left[\sum_{i=1}^n p_i\mathbf \Sigma p_i^*
-\sum_{i\ne j}p_i\mathbf \Sigma p_j^*\right]
=\sum_{i=1}^n\tr\Sigma_i
-\sum_{i\ne j}\tr[p_i\mathbf \Sigma p_j^*].
\]
Since the $\Sigma_i$'s are given, one equivalently seeks to minimise the last displayed double sum.  According to the lemma, the optimal $\mathbf \Sigma$ is the covariance operator of the random vector $(Y_1,\dots,Y_n)$ defined in the statement.


The probabilistic formulation highlights the interpretation of tangent space PCA in terms of functional data analysis, in particular in terms of the problem of \emph{phase variation} (or warping), and its solution, the process of \emph{registration} (or synchronisation).  Consider a situation where the variation of a random process $X$ arises via both \emph{amplitude} and \emph{phase} (see \citet[Section~2]{panaretos2016amplitude}):
\begin{enumerate}

\item First, one generates the realisation of a Gaussian process $X\sim N(0,\Sigma)$, viewed via the Karhunen--Lo\`eve expansion as
\[
X=\sum_{n=1}^{\infty}\sigma^{1/2}_n\xi_n\varphi_n
\]
for $\{\sigma_n,\varphi_n\}$ the eigenvalue/eigenfunction pairs of $\Sigma$, and $\xi\stackrel{iid}{\sim}N(0,1)$ an iid sequence of real standard Gaussian variables. This is the \emph{amplitude} variation layer, as corresponds to a superposition of random $N(0,\sigma_n)$ amplitude fluctuations  around fixed (deterministic) modes $\phi_n$.

\item Then, one warps the realisation $X$ into $\widetilde{X}$, by applying a positive definite operator $T$ (usually uncorrelated with $X$),
$$\widetilde{X}=TX=\sum_{n=1}^{\infty}\sigma^{1/2}_n\xi_nT\varphi_n$$
with the condition on $T$ that $\hs T\Sigma T\hs_1<\infty$, to guarantee that the resulting $\widetilde{X}$ has finite variance. This is the \emph{phase} variation layer, since it emanates from deformation fluctuations of the modes $\varphi_n$. The term \emph{phase} comes from the case $\mathcal{H}=L^2[0,1]$, where $\widetilde{X}(x)=(TX)(x)=\int_0^1 \tau(x,y)X(y)\mathrm{d}y=\sum_{n=1}^{\infty}\sigma^{1/2}_n\xi_n\int_0^1 \tau(x,y)\varphi_n(y)\mathrm{d}y$ can be seen to be variation attributable to the ``$x$-axis" (ordinate), contrasted to amplitude variation which is attributable to the ``$y$-axis" (abcissa).

\end{enumerate}

\noindent At the level of covariances, if $T$ is uncorrelated with $X$, then $\widetilde{X}$ has covariance $T\Sigma T$ conditional on $T$, which is a geodesic perturbation of $\Sigma$: it corresponds to the retraction (via the exponential map) of a \emph{linear perturbation} of $\Sigma$ on the tangent space $\Tan_\Sigma$. If this perturbation is ``zero mean" on the tangent space (i.e. $\mathbb{E}[T_k]=\mathscr{I}$), then one expects $\Sigma$ to be a Fr\'echet mean of the random operator $T\Sigma T$. So, if one gets to observe multiple such perturbations $\Sigma_k=T_k\Sigma T_k$, the tangent space PCA provides a means of \emph{registration} of $\{\Sigma_1,\dots,\Sigma_k\}$: the approximate recovery of $\Sigma$ and $\{T_k\}_{k=1}^n$, allowing for the separation of the amplitude from the phase variation (which, if left unaccounted, would have detrimental effects to statistical inference). This intuition is made precise in the next section, where phase variation is used as a means of suggesting a \emph{canonical generative model}: a statistical model behind the observed covariances $\{\Sigma_1,\dots,\Sigma_n\}$ that is naturally compatible with the use of the Procrustes distance.

Before moving on to this, we remark that, in a sense, we have come full circle.  The Procrustes distance of \citet{pigoli2014distances} is motivated by the Procrustes distance in shape theory, and is thus connected to the optimal simultaneous registration of multiple Euclidean point configurations, subjected to random isometries.  And, our interpretation of this distance, shows that it is connected to the optimal simultaneous registration of multiple Gaussian processes, subjected to random transportation deformations.

\section{Generative Models, Random Deformations, and Registration}\label{sec:generative}

An important question that has not yet been addressed regards the choice of the Procrustes distance for statistical purposes on covariance operators: \emph{why} would one choose this specific metric rather than another one? As the space of covariance operators is infinite-dimensional, there are naturally many other distances with which one can endow it.  For the statistician, the specific choice of metric on a space implicitly assumes a certain data generating mechanism for the sample at hand, and it is therefore of interest to ask what kind of generative model is behind the Procrustes distance. In the Introduction, we noticed that the use of a Hilbert-Schmidt distance on covariances implicitly postulates that second-order variation arises via \emph{additive perturbations},
$$\Sigma_k=\Sigma+E_k$$
for ${E}_k$ being zero mean self adjoint perturbations. Furnished with the insights of the optimal transportation perspective, particularly those gained in the last section (Section \ref{sec:pca}), we now show that the natural generative model associated with the Procrustes distance is one of \emph{random deformations} (a.k.a warping or phase variation), and is intimately related to the registration problem in functional data.  Suppose that $X$ is a Gaussian process with covariance $\Sigma$ and let $T$ be a random positive bounded operator on $\X$.  Conditional upon $T$, $TX$ is a Gaussian process with covariance $T\Sigma T^*$.  It is quite natural (and indeed necessary for identifiability) to assume that $\Sigma$ is the ``correct" (or \emph{template}) covariance, in the sense that the expected value of $T$ is the identity.  In other words, the covariance of $TX$ is the $\Sigma$ ``on average".  The conjugation perturbations 
$$\Sigma_k= T_k\Sigma T^*_k$$ 
then yield a generative model that is \emph{canonical} for the Procrustes metric, as we now rigorously show:
\begin{theorem}[Generative Model]\label{thm:identifiability}
Let $\Sigma$ be a covariance operator and let $T:\X\to\X$ be a random\footnote{In the sense that $T$ is Bochner measurable from a probability space $\Omega$ to $B(\X)$.  In particular, it is separately valued, namely $T(\Omega)$ is a separable subset of (the nonseparable) $B(\X)$.} positive linear map with $\mathbb E\hs T\hs_{\infty}^2<\infty$ and mean identity.  Then the random operator $T\Sigma T^*$ has $\Sigma$ as Fr\'echet mean in the Procrustes metric,
\[
\mathbb{E}[\Pi^2(\Sigma,T\Sigma T^*)]
\le
\mathbb{E}[\Pi^2(\Sigma',T\Sigma T^*)],
\]
for all non-negative nuclear operators $\Sigma'$.
\end{theorem}
The assumption that $\mathbb E\hs T\hs_{\infty}^2<\infty$ guarantees that the Fr\'echet functional $\E \Pi^2(A,T\Sigma T)$ is finite for any covariance (non-negative and nuclear) operator $\Sigma$. For measures on $\mathbb R^d$ with compact support, the result in Theorem~\ref{thm:identifiability} holds in a more general Wasserstein setup, where $\mu$ is a fixed measure and $T$ is a random optimal map with mean identity (\citet{bigot2012characterization}; \citet{zemel2017fr}).
\begin{proof}[Proof of Theorem \ref{thm:identifiability}]
We use the Kantorovich duality (\citet[Theorem 5.10]{villani2003topics}) as in Theorem~5 in \citet{zemel2017fr}.  Define the function $\varphi(x)=\langle Tx,x\rangle/2$ and its Legendre transform $\varphi^*(y)=\sup_{x\in \X} \langle x,y\rangle - \varphi(x)$.  We abuse notation for the interest of clarity, and write $d\Sigma(x)$ for integration with respect to the corresponding measure.  The strong and weak Kantorovich duality yield
\begin{align*}
W^2(N(0,\Sigma),N(0,T\Sigma T))
&=\ownint {\X}{}{\left(\frac12\|x\|^2 - \varphi(x)\right)}{\Sigma(x)}
+\ownint {\X}{}{\left(\frac12\|y\|^2 - \varphi^*(y)\right)}{T\Sigma T(x)}
;\\
W^2(N(0,\Sigma'),N(0,T\Sigma T))
&\ge\ownint {\X}{}{\left(\frac12\|x\|^2 - \varphi(x)\right)}{\Sigma'(x)}
+\ownint {\X}{}{\left(\frac12\|y\|^2 - \varphi^*(y)\right)}{T\Sigma T(x)}.
\end{align*}
Taking expectations, using Fubini's theorem and noting that $\E \varphi(x)=\|x\|^2/2$ because $\mathbb ET=\mathscr I$ formally proves the result; in particular this provides a proof for empirical Fr\'echet means (when $T$ takes finitely many values).

To make the calculations rigorous we modify the construction in \cite{zemel2017fr} to adapt for the unboundedness of the spaces.  Let $\Omega$ be the underlying probability space and $B(\mathcal H)$ the set of bounded operators on $\mathcal H$ with the operator norm topology.  We assume that $T:\Omega\to B(\X)$ is Bochner measurable with (Bochner) mean $\mathscr I$.  Then (the measure corresponding to) $T\Sigma T:\Omega\to \W(\X)$ is measurable because it is a (Lipschitz) continuous function of $T$.  To see this notice that
\[
W^2(N(0,S\Sigma S^*),N(0,T\Sigma T^*))
\le \ownint{\X}{}{\|S(x) - T(x)\|^2}{\Sigma(x)}
=\tr(S-T)\Sigma (S^*-T^*)
\le \hs S-T\hs_{\infty}^2\tr\Sigma.
\]
Similarly,
\[
\left|\ownint {\X}{}{\langle (T-S)x,x\rangle}{\Sigma'(x)}\right|
=\left|\tr(T-S)\Sigma'\right|
\le \hs T-S\hs_{\infty}\tr\Sigma'
\]
so the integrals with respect to $\varphi$ are measurable (from $\Omega$ to $\mathbb R$) for all $\Sigma'$, and integrable because $\mathbb E\hs T\hs_{\infty}<\infty$.  Since $W^2(\Sigma',T\Sigma T)$ is measurable and integrable, so are the integrals with respect to $\varphi^*$ (as a difference between integrable functionals).  To conclude the proof it remains to show that for all $\Sigma'$
\[
\E\ownint {\X}{}{ \langle Tx,x\rangle}{\Sigma'(x)}
=\ownint {\X}{}{ \langle (\E T)x,x\rangle}{\Sigma'(x)}.
\]
This is clearly true if $T$ is simple (takes finitely many values).  Otherwise, we can find a sequence of simple $T_n:\Omega\to B(\X)$ such that $\hs T_n - T\hs_\infty\to0$ almost surely and in expectation.  This Fubini equality holds for $T_n$ and
\begin{align*}
\left|\E\ownint {\X}{}{ \langle Tx,x\rangle}{\Sigma'(x)}
- \E\ownint {\X}{}{ \langle T_nx,x\rangle}{\Sigma'(x)}\right|
=|\E\tr(T-T_n)\Sigma'|
\le\tr\Sigma'\E \hs T - T_n\hs_{\infty};\\
\left|\ownint {\X}{}{ \langle (\E T)x,x\rangle}{\Sigma'(x)}
-\ownint {\X}{}{ \langle (\E T_n)x,x\rangle}{\Sigma'(x)}\right|
=|\tr (\E T - \E T_n)\Sigma'|
\le \tr\Sigma'\hs \E T - \E T_n\hs_{\infty}.
\end{align*}
By approximation the Fubini equality holds for $T$, completing the proof.
\end{proof}
The reader may have noticed that the proof relies on optimal transport arguments that do not make specific use of linearity of $T$ or Gaussianity. In the Gaussian case, however, the Fr\'echet functional can be evaluated explicitly due to the formula of the Wasserstein distance.  For the reader's convenience we outline another, more constructive proof of Theorem~\ref{thm:identifiability}.  The argument is fully rigorous in finite dimensions, and could probably be modified with additional effort to be valid in infinite dimensions.
\begin{proof}[Alternative proof of Theorem~\ref{thm:identifiability} in finite dimensions]
We first evaluate the Fr\'echet functional as
\begin{eqnarray*}
\mathbb{E}W^2(T\Sigma T^*,\Sigma)&=&\tr(\Sigma)+\mathbb{E}\tr(T\Sigma T^*)-2\mathbb{E}\tr(\Sigma^{1/2}T\Sigma T^*\Sigma^{1/2})^{1/2}\\
&=&\tr(\Sigma)+\tr( \mathbb{E}[T\otimes T]\Sigma)-2\mathbb{E}\tr(\Sigma^{1/2} T\Sigma^{1/2})\\
&=&\tr(\Sigma)+\tr((\mathscr{I}\otimes \mathscr{I}+\cov(T))\Sigma)-2\tr(\Sigma^{1/2} \mathbb{E}T\Sigma^{1/2})\\
&=&2\tr(\Sigma)+\tr(\cov(T)\Sigma)-2\tr(\Sigma)\\
&=&\tr(\cov(T)\Sigma_X ).
\end{eqnarray*}
We have used the fact that $T$ is self-adjoint, and that $\mathbb{E}T=\mathscr{I}$, so that $\mathbb{E}[T\otimes T]=\mathbb{E}T\otimes \mathbb{E}T+\mbox{cov}(T)=\mathscr{I}\otimes \mathscr{I}+\cov(T)$.  Keeping the above result in mind, we now compute the functional at an arbitrary $\Sigma'$:
\begin{eqnarray*}
\mathbb{E}W^2(T\Sigma T^*,\Sigma')
&=&\tr(\Sigma')+\mathbb{E}\tr(T\Sigma T^*)-2\mathbb{E}\tr(\Sigma'^{1/2} T\Sigma T^*\Sigma'^{1/2})^{1/2}\\
&=&
\tr(\cov(T)\Sigma)
+ \mathbb{E}\Big\{
 \tr(\Sigma'^{1/2}T\Sigma'^{1/2})+\tr(\Sigma^{1/2}T\Sigma^{1/2}) -2 \tr(\Sigma'^{1/2} T\Sigma T\Sigma'^{1/2})^{1/2}
  \Big\}.
\end{eqnarray*}
To prove that $F(\Sigma')\ge F(\Sigma)$ it suffices to show that the term inside the expectation is nonnegative;  we shall do this by interpreting it as the Wasserstein distance between $B=T^{1/2}\Sigma' T^{1/2}$ and $A=T^{1/2}\Sigma T^{1/2}$.  Write $B_1=\Sigma'^{1/2}T\Sigma'^{1/2}$, $A_1=\Sigma^{1/2}T\Sigma^{1/2}$.  Then the formula for the Wasserstein distance says that $2\tr (A^{1/2}BA^{1/2})^{1/2}\le \tr A + \tr B=\tr A_1 + \tr B_1$ (see \citet{dowson1982frechet}).  We thus only need to show that $\tr(A^{1/2}BA^{1/2})^{1/2}=\tr(\Sigma'^{1/2} T\Sigma T\Sigma'^{1/2})^{1/2}$.  Up until now, everything holds in infinite dimensions, but the next argument assumes finite dimensions.  In particular, will establish that $\tr(A^{1/2}BA^{1/2})^{1/2}=\tr(\Sigma'^{1/2} T\Sigma T\Sigma'^{1/2})^{1/2}$ by showing that these \emph{matrices} are conjugate.  Assume firstly that $\Sigma$, $\Sigma'$ and $T$ are invertible and write
\[
D = \Sigma'^{1/2} T\Sigma T\Sigma'^{1/2}
=\Sigma'^{-1/2}T^{-1/2}[BA]T^{1/2}\Sigma'^{1/2}
=\Sigma'^{-1/2}T^{-1/2}A^{-1/2}[A^{1/2}BA^{1/2}]A^{1/2}T^{1/2}\Sigma^{1/2}.
\]
Thus the positive matrices $D$ and $A^{1/2}BA^{1/2}$ have the same eigenvalues.  This also holds true for their square roots, that consequently have the same trace.  Since singular matrices can be approximated by nonsigular ones, this nonnegativity extends to the singular case as well, and so the proof is valid without restriction when $\X=\mathbb R^d$ is finite-dimensional.
%
\end{proof}

When the law of the random deformation $T$ is finitely supported,  we can get a stronger result than that of Theorem \ref{thm:identifiability}, without the assumption that $T$ is bounded.  This is not merely a technical improvement, as in many cases the optimal maps will, in fact, be unbounded (see the discussion after Equation \eqref{tangent_inner_product}).  This stronger result will also be used in Proposition~\ref{prop:fixedpoint} to obtain a (partial) characterisation of the sample Fr\'echet mean.
\begin{theorem}
\label{thm:identifiabilityEmpirical}
Let $\Sigma$ be a covariance operator corresponding to a centred Gaussian measure $\mu\equiv N(0,\Sigma)$ and let $D\subseteq \X$ be a dense linear subspace of $\mu$-measure one.  If $T_1,\dots,T_n:D\to\X$ are (possibly unbounded) linear operators such that $0\le\langle T_ix,x\rangle$, $\sum T_i(x)=nx$ and $\langle T_ix,y\rangle=\langle x,T_iy\rangle$ for all $i$ and all $x,y\in D$, then $\Sigma$ is a Fr\'echet mean of the finite collection $\{T_i\Sigma T_i:i=1,\dots,n\}$.
\end{theorem}
\begin{proof}
Straightforward calculations show that the functions $\varphi_i(x)=\langle T_ix,x\rangle/2$ are convex on $D$ and $T_ix$ is a subgradient of $\varphi_i$ for any $i$ and any $x\in D$.  The duality in the previous proof is therefore valid with the integrals involving $\varphi_i$ taken on $D$ (rather than on the whole of $\mathcal H$).  Since there are finitely many integrals, there are no measurability issues and we have $F(\mu)\le F(\nu)$ whenever $\nu(D)=1$.  By continuity considerations, since $D$ is dense in $\X$, this means that $F(\mu)\le F(\nu)$ for all $\nu\in \W(\X)$, so $\mu$, that is, $\Sigma$, is a Fr\'echet mean.
\end{proof}






\section{Characterisation of Fr\'echet Means via an Operator Equation}\label{sec:fixedpoint}
\citet{knott1984optimal} show that, in finite dimensions, a positive definite solution $\Sigma$ to the equation
\begin{equation}\label{eq:fixedpoint}
\Sigma
=\frac 1n\sum_{i=1}^n (\Sigma^{1/2} \Sigma_i\Sigma^{1/2})^{1/2}
\end{equation}
is a Fr\'echet mean of $\Sigma_1,\dots,\Sigma_n$  (see also \citet{ruschendorf2002n}.) Later, \citet{agueh2011barycenters} proved that \eqref{eq:fixedpoint} is in fact a characterisation of the mean in that (if one $\Sigma_i$ is invertible) this fixed point equation has a unique invertible solution, which is the Fr\'echet mean of $\Sigma_1,\dots,\Sigma_n$.  Part of their results extend easily to infinite dimensions:
\begin{proposition}
\label{prop:fixedpoint}
Let $\Sigma_1,\dots,\Sigma_n$ be covariance operators
.  Then:
\begin{enumerate}
\item Any Fr\'echet mean $\overline \Sigma$ of $\Sigma_1,\dots,\Sigma_n$ satisfies \eqref{eq:fixedpoint}.
\item If \eqref{eq:fixedpoint} holds, and $\mathrm{ker}(\Sigma)\subseteq \bigcap_{i=1}^{n} \mathrm{ker}(\Sigma_i)$, then $\Sigma$ is a Fr\'echet mean.
\end{enumerate}
\end{proposition}
\begin{proof}
Let $\mathscr P_k$ be a sequence of finite rank projections that converge strongly to the identity.  Then $\mathscr P_k\Sigma_i\mathscr P_k$ converge to $\Sigma_i$ in Wasserstein distance (Section~\ref{sec:projections}) and thus in trace norm by Proposition~\ref{prop:convEquiv}.  If $\overline \Sigma_k$ is the Fr\'echet mean of the projected operators, then it satisfies \eqref{eq:fixedpoint} by \cite[Theorem~6.1]{agueh2011barycenters}, with $\Sigma_i$ replaced by $\mathscr P_k\Sigma_i\mathscr P_k$.  But $\overline \Sigma_k$ converges to $\overline \Sigma$ in trace norm (Theorem~\ref{thm:convNfixed} and Proposition~\ref{prop:convEquiv}) and by continuity \eqref{eq:fixedpoint} holds true for $\overline \Sigma$.

Conversely, the kernel conditions means that $T_i=\topt\Sigma{\Sigma_i}=\Sigma^{-1/2}(\Sigma^{1/2}\Sigma_i\Sigma^{1/2})^{1/2}\Sigma^{-1/2}$ exists and is defined on a dense subspace $D_i$ of $\Sigma$-measure one (Proposition~\ref{prop:existenceMaps}).  Equation \eqref{eq:fixedpoint} yields $\sum T_i=n\mathscr I$ on $D=\cap_{i=1}^n D_i$, a set of full measure, and by Theorem~\ref{thm:identifiabilityEmpirical} $\Sigma$ is a Fr\'echet mean.
\end{proof}
In view of Proposition~\ref{prop:fixedpoint}, injectivity of the Fr\'echet mean is equivalent to the existence of an injective solution to the operator equation \eqref{eq:fixedpoint}.

\section{Open Questions and Future Work}\label{sec:future}

We conclude with some open questions. The most important of these is the injectivity (regularity) of the Fr\'echet mean $\overline \Sigma$.  We conjecture that, as in the finite dimensional case,
\begin{conjecture}[Regularity of the Fr\'echet Mean]\label{conj:regularity}
Let $\Sigma_1,\dots,\Sigma_n$ be covariances on $\mathcal{H}$ with $\Sigma_1$ injective.  Then, their Fr\'echet mean $\overline\Sigma$ with respect to the Procrustes metric $\Pi$ is also injective.
\end{conjecture}
Resolution of the conjecture (in the positive direction) will automatically yield the solution to the multicoupling problem (by virtue of Lemma~\ref{lem:frechetgivesmulti}) and the validity of the geodesic principal component analysis (Section~\ref{sec:pca}), irrespective of finite-dimensionality.  As mentioned in the previous section, injectivity is equivalent to the existence of an injective solution to the fixed point equation \eqref{eq:fixedpoint};  see Proposition~\ref{prop:fixedpoint}.

We remark that injectivity indeed holds when the operators in question commute, because then $\overline \Sigma^{1/2}=n^{-1}(\Sigma_1^{1/2},\dots,\Sigma_n^{1/2})$.  However, it is not possible to bound $\overline \Sigma^{1/2}$ in terms of $\Sigma_i^{1/2}$, not even up to constants.  More precisely,
\begin{lemma}
For all $\overline \Sigma$ of infinite rank and all $n\ge2$ there exist covariance operators $\Sigma_1,\dots,\Sigma_n$ with mean $\overline\Sigma$ and such for any positive number $c$,  $\overline \Sigma - c(\Sigma_1+\dots,\Sigma_n)$ is not positive.
\end{lemma}
\begin{proof}
Without loss of generality $n=2$, because if $\Sigma_1$ and $\Sigma_2$ are as required, then so are $(\Sigma_1,\Sigma_2,\overline\Sigma,\dots,\overline\Sigma)$. Let $\lambda_k$ and $\mu_k$ be disjoint sequences of nonzero eigenvalues of $\overline \Sigma$ and such that $\lambda_k/\mu_k>5^k$ (this is possible since the eigenvalues go to zero),  with corresponding eigenvectors $(e_k)$ and $(f_k)$.  Define an operator $T$ by $T(e_k)=e_k+b_kf_k$, $T(f_k)=b_ke_k+f_k$ (and $T$ is the identity on the orthogonal complement of the $e_k$'s and $f_k$'s).  Then $T$ is self-adjoint, and it is positive provided that $|b_k|\le1$.  We have for all $k$
\[
\frac{
\langle T\overline \Sigma Tf_k,f_k\rangle} {\langle \overline \Sigma f_k,f_k\rangle }
=\frac{\mu_k + b_k^2\lambda_k}{\mu_k}
\ge b_k^25^k
\to\infty
\]
if $b_k=2^{-k}$, say.  Therefore $\overline \Sigma - c(T\overline \Sigma T)$ is not positive for any $c>0$.  Also, $\hs T - \mathscr I\hs_\infty=1/2$ so $T$ has a bounded inverse and $2\mathscr I - T$ is also positive.  To complete the proof it suffices to see that $\overline \Sigma$ is the Fr\'echet mean of $\Sigma_1=T\overline \Sigma T$ and $\Sigma_2=(2\mathscr I-T)\overline \Sigma (2\mathscr I - T)$.  Indeed, the bounded operator $(2\mathscr I - T)\circ T^{-1}=2T^{-1} - \mathscr I$ is positive (as a composition of two positive operators that commute), and thus the optimal map $\topt12$ from $\Sigma_1$ to $\Sigma_2$.  The Fr\'echet mean is then the midpoint in McCann's interpolant, $[(\topt12+\mathscr I)/2]\#\Sigma_1=T^{-1}\#\Sigma_1=\overline \Sigma$.  (The Fr\'echet mean is unique here even if $\Sigma_i$ are not injective, since they have the same kernel and we can replace $\X$ by the orthogonal complement of this kernel to make them injective.)
\end{proof}
Note that $T$ (and $2\mathscr I - T$) can be very close to the identity in any Schatten norm, since we can multiply all the $b_k$'s by an arbitrary small constant.  It is therefore unlikely that Hajek--Feldman type conditions on the operators be relavant.

An alternative line of proof is by a variational argument.  Suppose that $\Sigma v=0$ for $\|v\|=1$ and define $\Sigma'=\Sigma + \epsilon v\otimes v$.  The quantity $\tr \Sigma'$ in the Fr\'echet functional increases by $\epsilon$, and we believe that $\tr (\Sigma_i \Sigma'\Sigma_i)^{1/2}$ behaves like $\sqrt\epsilon$ for $\epsilon$ small, and consequently $\Sigma'$ has a better Fr\'echet value.

Another important line of enquiry is the consistency of the (empirical) Fr\'echet mean of $\Sigma_1,\dots,\Sigma_n$ towards its population counterpart, as the sample size grows to infinity.  Under mild conditions on the law of $\Sigma$, this population mean is guaranteed to be unique (see Proposition~\ref{prop:uniqueFrechet} in Section~\ref{sec:existUnique}).  However, it is not known to exist in general;  the existence results of Le Gouic \& Loubes \cite{gouic2016existence} do not apply to $\X$ because the latter is not locally compact.  In view of Ziezold's \cite{ziezold1977expected} results, if a population mean exists and the sequence of empirical means converge, then the limit must be the population mean (under uniqueness).  These questions appear to be more subtle and we leave them for further work.

Finally, a last interesting question would be to establish the stability of the Procrustes algorithm to increasing projection dimension. In other words: if we had access to the fully infinite-dimensional covariances, it would still make sense to apply the Procrustes algorithm to obtain the Fr\'echet mean.  Would this still converge? The methods of proof of \citet{alvarez2016fixed} and \citet{zemel2017fr} are intrinsically finite dimensional, and cannot be lifted to infinite dimensions. Extending this convergence to the infinite dimensional case, would precisely establish the stability of the Procrustes algorithm to increasingly finer discretisations,  and this is likely to require new tools.  Preliminary simulation results indicate that the convergence is indeed quite stable to increasing the projection dimension, so we conjecture that the convergence result should be true.   In fact, this very issue may also lead to a resolution of Conjecture \ref{conj:regularity}, since one can show (as in \citet{zemel2017fr}) that the iterates of the algorithm stay injective at each step (provided the initial point and one of the $\Sigma_i$ is injective).

\section{Auxiliary results}\label{sec:aux}
\begin{lemma}
[trace of tensors]
\label{lem:tensornorm}
For all $f,g\in \X$ we have $\hs f\otimes g \hs_1=\|f\| \|g\|$.
\end{lemma}
\begin{proof}
This is clear if $g=0$.  Otherwise, since $f\otimes g$ is rank one, its trace norm equals its operator norm, which is
\[
\sup_{\|h\|=1}\|f\otimes gh\|
=\sup_{\|h\|=1}\|f\||\innprod gh|
=\|f\|\innprod g{g/\|g\|}|
=\|f\|\|g\|.
\]
\end{proof}
\begin{lemma}
[compact convex hulls]
\label{lem:compactHull}
Let $K$ be a compact subset of a Banach space.  Then its closed convex hull $\overline{\conv K}$ is compact.
\end{lemma}
\begin{proof}
We need to show that $\mathrm{Conv}K$ is totally bounded.  For any $\delta>0$ there exists a $\delta$-cover $x_1,\dots,x_n\in K$.  The simplex
\[
S
=
\mathrm{Conv}(x_1,\dots,x_n)
=
\left\{\sum_{i=1}^n a_i x_i:a_i\ge0,\sum_{i=1}^n a_i=1\right\}
\]
(a subset of $\mathrm{Conv}K$) is compact as a continuous image of the unit simplex in $\mathbb R^n$.  Indeed, if $v_i$ is the $i$-th column of the $n\times n$ identity matrix, then
\[
\omega\left(\sum_{i=1}^n a_iv_i\right)
=
\sum_{i=1}^n a_ix_i
\]
does the job.  Let $y_1,\dots,y_m$ be a $\delta$-cover of $S$ and let $y\in \mathrm{Conv}K$.  Then $y=\sum a_j z_j$ for $z_j\in K$ and $a_j\ge0$ that sum up to one.  For each $j$ there exists $i(j)$ such that $\|z_j - x_{i(j)}\|<\delta$.  Then $x=\sum a_j x_{i(j)}$ is in $S$ and $\|x - y\|<\delta$.  Thus, there exists $y_j$ such that $\|x - y_j\|<\delta$ and therefore $\|y - y_j\|<2\delta$ and total boundedness is established.
\end{proof}

\begin{lemma}[Gaussian fourth moment]\label{lem:Gauss4m}
Let $X$ centred Gaussian with covariance $\Sigma$.  Then
\[
\mathbb E\|X\|^4
\le
3(\mathbb E\|X\|^2)^2
=
3(\tr \Sigma)^2
.
\]
\end{lemma}
\begin{proof}
Let $(e_k)$ be a basis of eigenvectors of $\Sigma$ with eigenvalues $\lambda_k$.  Then $X_k=\langle X,e_k\rangle\sim N(0,\lambda_k)$ are independent and $\|X\|^2=\sum X_k^2$ has expectation $\sum \lambda_k=\mathrm{tr}\Sigma$.  Squaring gives
\[
\mathbb E \|X\|^4
=
\sum_{k,j}\mathbb E X_k^2X_j^2
=
\sum_k 3\lambda_k^2
+
\sum_{k\ne j} \lambda_k\lambda_j
=
\left(\sum \lambda_k\right)^2
+
2\sum \lambda_k^2
=
(\tr\Sigma)^2
+
2\hs \Sigma \hs_2^2
\le
3(\tr\Sigma)^2
.
\]
Interestingly, equality holds if and only if all but one eigenvalues are zero, i.e.\ $\Sigma$ has rank of at most one.
\end{proof}

\begin{corollary}\label{cor:Gauss4mom}
Let $\mathcal B\subset \mathcal W(\X)$ be a collection of centred Gaussian measures.  If $\mathcal B$ is tight and $\ownint {\X}{}{\|x\|^2}{\mu(x)}\le R$ for all $\mu\in \mathcal B$, then $\mathcal B$ is precompact in the Wasserstein space.
\end{corollary}
\begin{proof}
By Lemma~\ref{lem:Gauss4m} the fourth moment of all measures in $\mathcal B$ is bounded by $3R^2$.  Therefore
\[
\sup_{\mu\in\mathcal B}
\ownint{\|x\|>M}{}{\|x\|^2}{\mu(x)}
\le 
\frac 1{M^2}
\sup_{\mu\in\mathcal B}
\ownint{\|x\|>M}{}{\|x\|^4}{\mu(x)}
\le 
\frac {3R^2}{M^2}
\to 0
,\qquad 
M\to\infty.
\]
Any sequence $(\mu_n)\subseteq\mathcal B$ is tight and has a limit $\mu$ in distribution, and by Theorem~7.12 in \cite{villani2003topics}, $\mu$ is also a limit in $\mathcal W(\X)$.
\end{proof}

\begin{center}\textbf{Acknowledgements}\end{center}
Research supported in part by a Swiss National Science Foundation grant to V.~M. Panaretos.


\begin{thebibliography}{}

\bibitem[Agueh and Carlier, 2011]{agueh2011barycenters}
Agueh, M. and Carlier, G. (2011).
\newblock Barycenters in the {W}asserstein space.
\newblock {\em Society for Industrial and Applied Mathematics}, 43(2):904--924.

\bibitem[Alexander, 2005]{alexander2005multiple}
Alexander, D.~C. (2005).
\newblock Multiple-fiber reconstruction algorithms for diffusion mri.
\newblock {\em Annals of the New York Academy of Sciences}, 1064(1):113--133.

\bibitem[{\'A}lvarez-Esteban et~al., 2011]{alvarez2011uniqueness}
{\'A}lvarez-Esteban, P., Del~Barrio, E., Cuesta-Albertos, J., Matr{\'a}n, C.,
  et~al. (2011).
\newblock Uniqueness and approximate computation of optimal incomplete
  transportation plans.
\newblock In {\em Annales de l'Institut Henri Poincar{\'e}, Probabilit{\'e}s et
  Statistiques}, volume~47, pages 358--375. Institut Henri Poincar{\'e}.

\bibitem[{\'A}lvarez-Esteban et~al., 2016]{alvarez2016fixed}
{\'A}lvarez-Esteban, P.~C., del Barrio, E., Cuesta-Albertos, J., and
  Matr{\'a}n, C. (2016).
\newblock A fixed-point approach to barycenters in {W}asserstein space.
\newblock {\em Journal of Mathematical Analysis and Applications},
  441(2):744--762.

\bibitem[Ambrosio and Gigli, 2013]{ambrosio2013user}
Ambrosio, L. and Gigli, N. (2013).
\newblock A user’s guide to optimal transport.
\newblock In {\em Modelling and optimisation of flows on networks}, pages
  1--155. Springer.

\bibitem[Ambrosio et~al., 2008]{ambrosio2008gradient}
Ambrosio, L., Gigli, N., and Savar{\'e}, G. (2008).
\newblock {\em Gradient flows: in metric spaces and in the space of probability
  measures}.
\newblock Springer Science \& Business Media.

\bibitem[Benko et~al., 2009]{benko:2009}
Benko, M., H{\"a}rdle, W., and Kneip, A. (2009).
\newblock {Common functional principal components.}
\newblock {\em Ann. Stat.}, 37(1):1--34.

\bibitem[Bhattacharya and Patrangenaru, 2003]{bhattacharya2003large}
Bhattacharya, R. and Patrangenaru, V. (2003).
\newblock Large sample theory of intrinsic and extrinsic sample means on
  manifolds: i.
\newblock {\em Annals of statistics}, pages 1--29.

\bibitem[Bhattacharya and Patrangenaru, 2005]{bhattacharya2005large}
Bhattacharya, R. and Patrangenaru, V. (2005).
\newblock Large sample theory of intrinsic and extrinsic sample means on
  manifolds: ii.
\newblock {\em Annals of statistics}, pages 1225--1259.

\bibitem[Bigot and Klein, 2012]{bigot2012characterization}
Bigot, J. and Klein, T. (2012).
\newblock Characterization of barycenters in the {W}asserstein space by
  averaging optimal transport maps.
\newblock {\em arXiv preprint arXiv:1212.2562}.

\bibitem[Bogachev, 1998]{bogachev1998gaussian}
Bogachev, V.~I. (1998).
\newblock {\em Gaussian measures}, volume~62.
\newblock American Mathematical Society Providence.

\bibitem[Brenier, 1991]{brenier1991polar}
Brenier, Y. (1991).
\newblock Polar factorization and monotone rearrangement of vector-valued
  functions.
\newblock {\em Communications on pure and applied mathematics}, 44(4):375--417.

\bibitem[Coffey et~al., 2011]{coffey2011common}
Coffey, N., Harrison, A., Donoghue, O., and Hayes, K. (2011).
\newblock Common functional principal components analysis: A new approach to
  analyzing human movement data.
\newblock {\em Human movement science}, 30(6):1144--1166.

\bibitem[Cuesta-Albertos et~al., 1996]{cuesta1996lower}
Cuesta-Albertos, J., Matr{\'a}n-Bea, C., and Tuero-Diaz, A. (1996).
\newblock On lower bounds for the $l_2$-{W}asserstein metric in a {H}ilbert
  space.
\newblock {\em Journal of Theoretical Probability}, 9(2):263--283.

\bibitem[Cuesta-Albertos and Matr{\'a}n, 1989]{cuesta1989notes}
Cuesta-Albertos, J.~A. and Matr{\'a}n, C. (1989).
\newblock Notes on the {W}asserstein metric in {H}ilbert spaces.
\newblock {\em The Annals of Probability}, 17(3):1264--1276.

\bibitem[Descary and Panaretos, 2016]{descary2016functional}
Descary, M.-H. and Panaretos, V.~M. (2016).
\newblock Functional data analysis by matrix completion.
\newblock {\em arXiv preprint arXiv:1609.00834}.

\bibitem[Dowson and Landau, 1982]{dowson1982frechet}
Dowson, D. and Landau, B. (1982).
\newblock The {F}r{\'e}chet distance between multivariate normal distributions.
\newblock {\em Journal of multivariate analysis}, 12(3):450--455.

\bibitem[Dryden and Mardia, 1998]{dryden1998statistical}
Dryden, I. and Mardia, K. (1998).
\newblock {\em Statistical analysis of shape}.
\newblock Wiley.

\bibitem[Dryden et~al., 2009]{dryden2009non}
Dryden, I.~L., Koloydenko, A., and Zhou, D. (2009).
\newblock Non-euclidean statistics for covariance matrices, with applications
  to diffusion tensor imaging.
\newblock {\em The Annals of Applied Statistics}, pages 1102--1123.

\bibitem[Durrett, 2010]{durrett2010probability}
Durrett, R. (2010).
\newblock {\em Probability: theory and examples}.
\newblock Cambridge university press.

\bibitem[Fletcher et~al., 2004]{fletcher2004principal}
Fletcher, P.~T., Lu, C., Pizer, S.~M., and Joshi, S. (2004).
\newblock Principal geodesic analysis for the study of nonlinear statistics of
  shape.
\newblock {\em IEEE transactions on medical imaging}, 23(8):995--1005.

\bibitem[Fr{\'e}chet, 1948]{frechet1948elements}
Fr{\'e}chet, M. (1948).
\newblock Les {\'e}l{\'e}ments al{\'e}atoires de nature quelconque dans un
  espace distanci{\'e}.
\newblock {\em Ann. Inst. H. Poincar{\'e}}, 10(3):215--310.

\bibitem[Fremdt et~al., 2013]{fremdt2013testing}
Fremdt, S., Steinebach, J.~G., Horv{\'a}th, L., and Kokoszka, P. (2013).
\newblock Testing the equality of covariance operators in functional samples.
\newblock {\em Scandinavian Journal of Statistics}, 40(1):138--152.

\bibitem[Gabrys et~al., 2010]{gabrys2010tests}
Gabrys, R., Horv{\'a}th, L., and Kokoszka, P. (2010).
\newblock Tests for error correlation in the functional linear model.
\newblock {\em Journal of the American Statistical Association},
  105(491):1113--1125.

\bibitem[Gangbo and Swiech, 1998]{gangbo1998optimal}
Gangbo, W. and Swiech, A. (1998).
\newblock Optimal maps for the multidimensional monge-kantorovich problem.
\newblock {\em Communications on pure and applied mathematics}, 51(1):23--45.

\bibitem[Gower, 1975]{gower1975generalized}
Gower, J.~C. (1975).
\newblock Generalized {P}rocrustes analysis.
\newblock {\em Psychometrika}, 40(1):33--51.

\bibitem[Horv{\'a}th et~al., 2013]{horvath2013test}
Horv{\'a}th, L., Hu{\v{s}}kov{\'a}, M., and Rice, G. (2013).
\newblock Test of independence for functional data.
\newblock {\em Journal of Multivariate Analysis}, 117:100--119.

\bibitem[Horv{\'a}th and Kokoszka, 2012]{horvath2012inference}
Horv{\'a}th, L. and Kokoszka, P. (2012).
\newblock {\em Inference for functional data with applications}, volume 200.
\newblock Springer Science \& Business Media.

\bibitem[Hsing and Eubank, 2015]{hsing2015theoretical}
Hsing, T. and Eubank, R. (2015).
\newblock {\em Theoretical foundations of functional data analysis, with an
  introduction to linear operators}.
\newblock John Wiley \& Sons.

\bibitem[Huckemann et~al., 2010]{huckemann2010intrinsic}
Huckemann, S., Hotz, T., and Munk, A. (2010).
\newblock Intrinsic shape analysis: Geodesic pca for {R}iemannian manifolds
  modulo isometric {L}ie group actions.
\newblock {\em Statistica Sinica}, pages 1--58.

\bibitem[Jaru{\v{s}}kov{\'a}, 2013]{jaruvskova2013testing}
Jaru{\v{s}}kov{\'a}, D. (2013).
\newblock Testing for a change in covariance operator.
\newblock {\em Journal of Statistical Planning and Inference},
  143(9):1500--1511.

\bibitem[Jolliffe, 2002]{jolliffe2002principal}
Jolliffe, I.~T. (2002).
\newblock Principal component analysis and factor analysis.
\newblock {\em Principal component analysis}, pages 150--166.

\bibitem[Karcher, 1977]{karcher1977riemannian}
Karcher, H. (1977).
\newblock Riemannian center of mass and mollifier smoothing.
\newblock {\em Communications on pure and applied mathematics}, 30(5):509--541.

\bibitem[Knott and Smith, 1984]{knott1984optimal}
Knott, M. and Smith, C.~S. (1984).
\newblock On the optimal mapping of distributions.
\newblock {\em Journal of Optimization Theory and Applications}, 43(1):39--49.

\bibitem[Kraus, 2014]{kraus2014components}
Kraus, D. (2014).
\newblock Components and completion of partially observed functional data.
\newblock {\em Journal of the Royal Statistical Society: Series B (Statistical
  Methodology)}.

\bibitem[Kraus and Panaretos, 2012]{panaretos_biometrika}
Kraus, D. and Panaretos, V.~M. (2012).
\newblock Dispersion operators and resistant second-order functional data
  analysis.
\newblock {\em Biometrika}, 99(4):813--832.

\bibitem[Le~Gouic and Loubes, 2016]{gouic2016existence}
Le~Gouic, T. and Loubes, J.-M. (2016).
\newblock Existence and consistency of {W}asserstein barycenters.
\newblock {\em Probability Theory and Related Fields}, pages 1--17.

\bibitem[McCann, 1997]{mccann1997convexity}
McCann, R.~J. (1997).
\newblock A convexity principle for interacting gases.
\newblock {\em Advances in mathematics}, 128(1):153--179.

\bibitem[Olkin and Pukelsheim, 1982]{olkin1982distance}
Olkin, I. and Pukelsheim, F. (1982).
\newblock The distance between two random vectors with given dispersion
  matrices.
\newblock {\em Linear Algebra and its Applications}, 48:257--263.

\bibitem[Panaretos et~al., 2010]{panaretos_jasa}
Panaretos, V.~M., Kraus, D., and Maddocks, J.~H. (2010).
\newblock Second-order comparison of gaussian random functions and the geometry
  of dna minicircles.
\newblock {\em J. Amer. Statist. Assoc.}, 105(490):670--682.

\bibitem[Panaretos and Tavakoli, 2013]{panaretos2013cramer}
Panaretos, V.~M. and Tavakoli, S. (2013).
\newblock Cram{\'e}r--karhunen--lo{\`e}ve representation and harmonic principal
  component analysis of functional time series.
\newblock {\em Stochastic Processes and their Applications}, 123(7):2779--2807.

\bibitem[Panaretos and Zemel, 2016]{panaretos2016amplitude}
Panaretos, V.~M. and Zemel, Y. (2016).
\newblock Amplitude and phase variation of point processes.
\newblock {\em The Annals of Statistics}, 44(2):771--812.

\bibitem[Panaretos and Zemel, pear]{panaretos2018introduction}
Panaretos, V.~M. and Zemel, Y. (2018 (to appear)).
\newblock {\em Introduction to Statistics in the {W}asserstein Space}.
\newblock Springer Briefs in Probability and Mathematical Statistics.

\bibitem[Paparoditis and Sapatinas, 2014]{paparoditis2014bootstrap}
Paparoditis, E. and Sapatinas, T. (2014).
\newblock Bootstrap-based testing for functional data.
\newblock {\em arXiv preprint arXiv:1409.4317}.

\bibitem[Pigoli et~al., 2014]{pigoli2014distances}
Pigoli, D., Aston, J.~A., Dryden, I.~L., and Secchi, P. (2014).
\newblock Distances and inference for covariance operators.
\newblock {\em Biometrika}, 101(2):409--422.

\bibitem[Ramsay and Silverman, 2005]{ramsay2005springer}
Ramsay, J. and Silverman, B. (2005).
\newblock Springer series in statistics.

\bibitem[R{\"u}schendorf and Rachev, 1990]{ruschendorf1990characterization}
R{\"u}schendorf, L. and Rachev, S.~T. (1990).
\newblock A characterization of random variables with minimum {$L^2$}-distance.
\newblock {\em Journal of Multivariate Analysis}, 32(1):48--54.

\bibitem[R{\"u}schendorf and Uckelmann, 2002]{ruschendorf2002n}
R{\"u}schendorf, L. and Uckelmann, L. (2002).
\newblock On the n-coupling problem.
\newblock {\em Journal of multivariate analysis}, 81(2):242--258.

\bibitem[Schwartzman, 2006]{schwartzman2006random}
Schwartzman, A. (2006).
\newblock {\em Random ellipsoids and false discovery rates: Statistics for
  diffusion tensor imaging data}.
\newblock PhD thesis, Stanford University.

\bibitem[Schwartzman et~al., 2008]{schwartzman2008false}
Schwartzman, A., Dougherty, R.~F., and Taylor, J.~E. (2008).
\newblock False discovery rate analysis of brain diffusion direction maps.
\newblock {\em The Annals of Applied Statistics}, pages 153--175.

\bibitem[Stein and Shakarchi, 2009]{stein2009real}
Stein, E.~M. and Shakarchi, R. (2009).
\newblock {\em Real analysis: measure theory, integration, and {H}ilbert
  spaces}.
\newblock Princeton University Press.

\bibitem[Takatsu, 2011]{takatsu2011wasserstein}
Takatsu, A. (2011).
\newblock Wasserstein geometry of {G}aussian measures.
\newblock {\em Osaka Journal of Mathematics}, 48(4):1005--1026.

\bibitem[Tavakoli and Panaretos, 2016]{tavakoli:2014}
Tavakoli, S. and Panaretos, V.~M. (2016).
\newblock Detecting and localizing differences in functional time series
  dynamics: a case study in molecular biophysics.
\newblock {\em Journal of the American Statistical Association},
  111(515):1020--1035.

\bibitem[Van Der~Vaart and Wellner, 1996]{van1996weak}
Van Der~Vaart, A.~W. and Wellner, J.~A. (1996).
\newblock Weak convergence.
\newblock In {\em Weak Convergence and Empirical Processes}, pages 16--28.
  Springer.

\bibitem[Villani, 2003]{villani2003topics}
Villani, C. (2003).
\newblock {\em Topics in Optimal Transportation}, volume~58.
\newblock American Mathematical Society.

\bibitem[Wang et~al., 2016]{wang2016functional}
Wang, J.-L., Chiou, J.-M., and M{\"u}ller, H.-G. (2016).
\newblock Functional data analysis.
\newblock {\em Annual Review of Statistics and Its Application}, 3:257--295.

\bibitem[Yao et~al., 2005a]{yao2005functional}
Yao, F., M{\"u}ller, H.-G., and Wang, J.-L. (2005a).
\newblock Functional data analysis for sparse longitudinal data.
\newblock {\em Journal of the American Statistical Association},
  100(470):577--590.

\bibitem[Yao et~al., 2005b]{yao2005functional2}
Yao, F., M{\"u}ller, H.-G., Wang, J.-L., et~al. (2005b).
\newblock Functional linear regression analysis for longitudinal data.
\newblock {\em The Annals of Statistics}, 33(6):2873--2903.

\bibitem[Zemel, 2017]{zemel2017thesis}
Zemel, Y. (2017).
\newblock {\em Fr{\'e}chet means in {W}asserstein space: theory and
  algorithms}.
\newblock PhD thesis, \'Ecole Polytechnique F\'ed\'erale de Lausanne.

\bibitem[Zemel and Panaretos, 2017]{zemel2017fr}
Zemel, Y. and Panaretos, V.~M. (2017).
\newblock Fr{\'e}chet means and {P}rocrustes analysis in {W}asserstein space.
\newblock {\em Bernoulli (to appear), available on arXiv:1701.06876}.

\bibitem[Zhang, 2013]{zhang:2013}
Zhang, J. (2013).
\newblock {\em Analysis of Variance for Functional Data}.
\newblock Monographs on statistics and applied probability. Chapman \& Hall.

\bibitem[Ziezold, 1977]{ziezold1977expected}
Ziezold, H. (1977).
\newblock On expected figures and a strong law of large numbers for random
  elements in quasi-metric spaces.
\newblock In {\em Transactions of the Seventh Prague Conference on Information
  Theory, Statistical Decision Functions, Random Processes and of the 1974
  European Meeting of Statisticians}, pages 591--602. Springer.

\end{thebibliography}
\end{document}